\newif\ifsigconf
\sigconftrue
\sigconffalse   

\ifsigconf
\documentclass[sigconf]{acmart}
\copyrightyear{2020}
\acmYear{2020}
\setcopyright{acmcopyright}
\acmConference[SPAA '20]{Proceedings of the 32nd ACM
Symposium on Parallelism in Algorithms and Architectures}{July 15--17, 2020}{Virtual
Event, USA}
\acmBooktitle{Proceedings of the 32nd ACM Symposium on Parallelism in Algorithms
and Architectures (SPAA '20), July 15--17, 2020, Virtual Event, USA}
\acmPrice{15.00}
\acmDOI{10.1145/3350755.3400218}
\acmISBN{978-1-4503-6935-0/20/07}

\else
\documentclass[]{article}
\usepackage{fullpage}
\fi

\usepackage[inline]{enumitem}   
\usepackage{amsmath}            
\usepackage{amsthm}             
\usepackage{xspace}             
\usepackage{amsfonts}           
\usepackage{graphicx}
\usepackage{changepage}         
\usepackage{tikz}

\newtheorem{theorem}{Theorem}
\newtheorem{lemma}{Lemma}[section]
\newtheorem{remark}[lemma]{Remark}

\newtheorem{definition}[lemma]{Definition}

\newcommand{\eps}{\varepsilon}
\newcommand{\varmatrix}[1]{\mathbf{#1}}
\newcommand{\ignore}[1]{}

\newcommand{\free}[1]{R(#1)}
\newcommand{\release}[1]{r_{#1}}
\newcommand{\size}[1]{d_{#1}}
\newcommand{\completion}[1]{C_{#1}}
\newcommand{\response}[1]{\rho_{#1}} 

\newcommand{\capacity}[1]{c_{#1}}
\newcommand{\mincapacity}[1]{\kappa_{#1}}
\newcommand{\flows}[1]{F_{#1}}
\newcommand{\schedule}[1]{\sigma_{#1}}
\newcommand{\instance}[1]{S_{#1}}
\newcommand{\fresponse}[1]{\Delta_{#1}}

\newcommand{\program}[2]{(#1) - (#2)}

\newcommand{\fsart}{{\sc FS-ART}}
\newcommand{\fsmrt}{{\sc FS-MRT}}

\newif\iflong
\longtrue

\newcommand{\ListLengths}{\setlength{\itemsep}{0ex}\setlength{\topsep}{1ex}\setlength{\partopsep}{0ex}}
\newenvironment{mylowitemize}{\begin{list}{$\bullet$}{\setlength{\leftmargin}{1em}\ListLengths}}{\end{list}}

\newif\ifnotes\notesfalse

\ifnotes

\newcommand{\david}[1]{\textcolor{cyan}{\textrm{[David says: #1]}}}
\newcommand{\rajmohan}[1]{\textcolor{red}{\textrm{[Rajmohan says: #1]}}}

\else

\newcommand{\david}[1]{}
\newcommand{\rajmohan}[1]{}

\fi

\begin{document}

\title{Scheduling Flows on a Switch to Optimize Response Times}

\ifsigconf
\else

\author{
    Hamidreza Jahanjou \thanks{Google, Mountain View, CA, USA. Email: \texttt{hamidrj@google.com}}
\and    
    Rajmohan Rajaraman\thanks{Northeastern University, Boston, MA, USA. Email: \texttt{rraj@ccs.neu.edu}} 
\and
    David Stalfa\thanks{Northeastern University, Boston, MA, USA. Email: \texttt{stalfa@ccis.neu.edu}}
}
\date{}
\maketitle
\begin{abstract}
    We study the scheduling of flows on a switch with the goal of optimizing metrics related to the response time of the flows. The input is a sequence of flow requests on a switch, where the switch is represented by a bipartite graph with a capacity on each vertex (port), and a flow request is an edge with associated demand. In each round, a subset of edges can be scheduled under the constraint that the total demand of the scheduled edges incident on any vertex is at most the capacity of the vertex.  This class of scheduling problems has applications in datacenter networks, and has been extensively studied. Previous work has essentially settled the complexity of metrics based on {\em completion time}.  The objective of average or maximum {\em response time}, however, is more challenging. \ifnotes 
To the best of our knowledge, there are no prior approximation algorithms results for these metrics in the context of flow scheduling.
\fi 

We present approximation algorithms for flow scheduling over a switch to optimize response time based metrics.  For the average response time metric, whose NP-hardness follows directly from past work, we present an offline $O(1 + O(\log(n))/c)$ approximation algorithm for unit flows, assuming that the port capacities of the switch can be increased by a factor of $1 + c$, for any given positive integer $c$.  For the maximum response time metric, we first establish that it is NP-hard to achieve an approximation factor of better than 4/3 without augmenting capacity.  We then present an offline algorithm that achieves {\em optimal maximum response time}, assuming the capacity of each port is increased by at most $2 d_{max} - 1$, where $d_{max}$ is the maximum demand of any flow.  Both algorithms are based on linear programming relaxations.  We also study the online version of flow scheduling using the lens of competitive analysis, and present preliminary results along with experiments that evaluate the performance of fast online heuristics. 
\end{abstract}
\newpage

\fi





\section{Introduction}
With the advent of software-defined networking (SDN) and OpenFlow switch protocol, routing and scheduling in modern data center networks is increasingly performed at the level of flows. A \emph{flow} is a particular set of application traffic between two endpoints that receive the same forwarding decisions. As a consequence of the shift towards centralized flow-based control, efficient algorithms for scheduling and routing of flows and their variants have gained prominent importance \cite{chowdhury+kpyy:coflow,QiuSteinZhong2015,sdn,sdn2,sdnbook}.

In order to model the datacenter network carrying the flows, it is common to represent the entire datacenter network as one non-blocking switch (see Figure \ref{fig:switch}) interconnecting all machines \cite{alizadeh,ballani,kang,QiuSteinZhong2015}. This simple model is attractive because of advances in full-bisection bandwidth topologies \cite{vl2,Niran}. In this model, every input (ingress) port is connected to every output (egress) port. Bandwidth limits are at the ports and the interconnections are assumed to have unlimited bandwidth.  We model the datacenter network as a general bipartite graph (which includes the full-bisection as a special case) with capacities at each vertex (port).

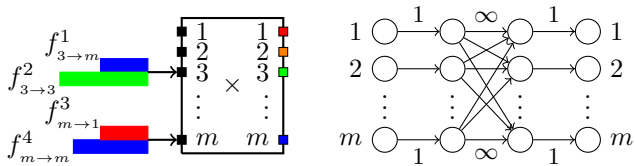
\begin{figure}[b]
\centering
\begin{tikzpicture}[scale=0.9]
\draw[thick] (0,0) -- (1.5,0) -- (1.5,2) -- (0,2) -- (0,0); 
\node at (0.25,0.8) {$\vdots$};
\node[draw, label=right:1, fill, scale=0.5] (l1) at (0,1.8) {};
\node[draw, label=right:2, fill, scale=0.5] (l2) at (0,1.5) {};
\node[draw, label=right:3, fill, scale=0.5] (l3) at (0,1.2) {};
\node[draw, label=right:$m$, fill, scale=0.5] (lm) at (0,0.2) {};
\node at (1.25,0.8) {$\vdots$};
\node[draw, label=left:1, fill=red, scale=0.5] (r1) at (1.5,1.8) {};
\node[draw, label=left:2, fill=orange, scale=0.5] (r2) at (1.5,1.5) {};
\node[draw, label=left:3, fill=green, scale=0.5] (r3) at (1.5,1.2) {};
\node[draw, label=left:$m$, fill=blue, scale=0.5] (rm) at (1.5,0.2) {};

\node at (0.75,1) {$\times$};

\draw[fill=blue,blue] (-1.2,1.2) rectangle (-0.5,1.4);
\draw[fill=orange,green] (-1.8,1) rectangle (-0.5,1.2);
\draw [->,thick] (-0.5,1.2) -- (l3);
\node at (-1.6,1.6) {$f^1_{_{3\rightarrow m}}$};
\node at (-2.2,1.1) {$f^2_{_{3\rightarrow 3}}$};

\draw[fill=orange,blue] (-1.6,0) rectangle (-0.5,0.2);
\draw[fill=blue,red] (-1.2,0.2) rectangle (-0.5,0.4);
\draw [->,thick] (-0.5,0.2) -- (lm);
\node at (-1.6,0.6) {$f^3_{_{m\rightarrow 1}}$};
\node at (-2.1,0.1) {$f^4_{_{m\rightarrow m}}$};
\node[draw,circle,label=left:1] (nl1) at (3,1.8) {};
\node[draw,circle,label=left:2] (nl2) at (3,1.25) {};
\node (nl3) at (3,0.8) {$\vdots$};
\node[draw,circle,label=left:$m$] (nl4) at (3,0.2) {};

\node[draw,circle] (pl1) at (4,1.8) {};
\node[draw,circle] (pl2) at (4,1.25) {};
\node (pl3) at (4,0.8) {$\vdots$};
\node[draw,circle] (pl4) at (4,0.2) {};

\node[draw,circle] (pr1) at (5,1.8) {};
\node[draw,circle] (pr2) at (5,1.25) {};
\node (pr3) at (5,0.8) {$\vdots$};
\node[draw,circle] (pr4) at (5,0.2) {};

\node[draw,circle,label=right:1] (nr1) at (6,1.8) {};
\node[draw,circle,label=right:2] (nr2) at (6,1.25) {};
\node (nr3) at (6,0.8) {$\vdots$};
\node[draw,circle,label=right:$m$] (nr4) at (6,0.2) {};

\draw[->] (pl1) -- node [midway,above] {\small$\infty$} (pr1);
\draw[->] (pl1) -- (pr2);
\draw[->] (pl1) -- (pr4);

\draw[->] (pl2) -- (pr1);
\draw[->] (pl2) -- (pr2);
\draw[->] (pl2) -- (pr4);

\draw[->] (pl4) -- (pr1);
\draw[->] (pl4) -- (pr2);
\draw[->] (pl4) -- node [midway,below] {\small$\infty$} (pr4);

\draw[->] (nl1) -- node [midway,above] {\small$1$} (pl1);
\draw[->] (nl2) -- (pl2);
\draw[->] (nl4) -- node [midway,below] {\small$1$} (pl4);

\draw[->] (pr1) -- node [midway,above] {\small$1$} (nr1);
\draw[->] (pr2) -- (nr2);
\draw[->] (pr4) -- node [midway,below] {\small$1$} (nr4);
\end{tikzpicture}
\caption{\label{fig:switch}(left) An $m\times m$ non-blocking switch with unit port capacities. Each incoming flow is shown as a bar on the left, with the length of the bar proportionate to the flow size. Each flow also specifies its input and output ports. 
For instance, two flows $f^1$ and $f^4$ share the same destination port.  
(right) The switch can be regarded as a complete $m\times m$ bipartite graph augmented with two sets of parallel edges.}
\end{figure}

In the context of scheduling and client-server applications,
{\em response time}--also known as flow time or sojourn time--is a very natural and important objective.  Indeed, response time is directly related to quality of service experienced by clients \cite{bansal_thesis,whyresponsetime}.  In the job scheduling literature, metrics related to response times have been extensively studied in diverse frameworks, including approximation algorithms~\cite{Bansal2015flowtime,batra+gk:flow,chekuri+kz:schedule,feige+kl:flow,kellerer+tw:schedule}, competitive analysis~\cite{bansal+chan,kalyanasundaram+p:schedule_nonclairvoyant,Mastrolilli2003}, and queuing-theoretic analysis~\cite{biersack+su:schedule,grosof+sh:srpt}.  For flow scheduling, however, response time optimization is not as well-understood as completion time optimization; to the best of our knowledge, there is no prior work on approximation algorithms for flow scheduling to optimize response time metrics.  In this paper, we study the problem of scheduling flows on a switch network to minimize average response time and maximum response time.

    \subsection{Results}
    We present approximation algorithms for flow scheduling on a bipartite switch network to minimize response time metrics.
\begin{itemize}[leftmargin=5mm]
    \item We present a $(1 + c, O(\log n/c))$-approximation algorithm, running in polynomial time, for scheduling $n$ {\em unit flows}\/
    under the average response time metric, for any given positive integer $c$; that is, our algorithm achieves an average response time of $O(\log n)/c$ times the optimal assuming it is allowed port capacity that is $1 + c$ times that of the original.  Our results on average response time appear in Section~\ref{sec:avgresponsetime}.
    \item We show that it is NP-hard to attain an approximation factor smaller than $4/3$ for the maximum response time metric.  We next present a polynomial-time algorithm that achieves {\em optimal}\/ maximum response time, assuming it is allowed port capacity that is at most $2d_{max} - 1$ more than that of the optimal, where $d_{max}$ is the maximum demand of any flow request.  For the special case of unit demands, note that this is best possible, given the hardness result.  Our results on maximum response time appear in Section~\ref{sec:maxresponsetime}.
\end{itemize}
Both of our algorithms are based on rounding a suitable linear programming relaxation of the associated problem.  The algorithm for average response time uses the iterative rounding paradigm, along the lines of previous work in scheduling jobs on unrelated machines~\cite{Bansal2015flowtime}.  A challenge we need to address is that a "job" in flow scheduling uses two different capacitated "resources" (ports) simultaneously.  We are able to overcome this challenge if we allow resource augmentation.  An important open problem is to determine whether polylogarithmic- or better approximations for average response time are achievable without resource augmentation.

For maximum response time, our hardness reduction is through the classic Timetable problem~\cite{timetable} and provides a useful target for practitioners developing heuristics.  Our approximation algorithm is achieved by applying a rounding theorem of~\cite{Karp87globalwire}, and in fact extends to the more general problem in which we need to meet {\em distinct deadlines}\/ for individual flows\iflong, as opposed to a uniform maximum response time\fi.

Both the algorithms above are offline approximations.  In Section~\ref{sec:online}, we study online algorithms for response time metrics.
\begin{itemize}[leftmargin=5mm]
    \item We present preliminary theoretical results including a resource-augmented constant-factor competitive algorithm for maximum response time, which builds on our offline algorithm.  We next present experimental evaluations of natural online heuristics for average and maximum response time metrics.
\end{itemize}
Our work leaves some intriguing open problems and several directions for future research, which are highlighted in Section~\ref{sec:open}.
    
    \subsection{Related Work}
    There is considerable work on scheduling flows on non-blocking switch networks as well as more general topologies, primarily for completion time metrics.  There is extensive literature on scheduling matchings over high-speed crossbar switches; these studies largely adopt a queuing theoretic framework (e.g., see~\cite{giaccone+ps:switch,gong+tlyx:switch,shah+s:schedule}). In \cite{ChowdhuryZhongStoica2014}, Chowdhury et al. present effective heuristics for scheduling generalizations of flows, called co-flows, without release times on a non-blocking switch network. 
More recently, Luo et al. \cite{luo+foerster+etal:splitcastreconfig} provide heuristics for scheduling multicast flows over a reconfigurable switch.
Approximation algorithms for average completion time of co-flows on a non-blocking switch are given in~\cite{Ahmadi2017,KM-coflow-SPAA16,shafiee+g:coflow,DBLP:conf/spaa/QiuSZ15}.  Scheduling over general network topologies is studied in~\cite{chowdhury+kpyy:coflow,Jahanjou_spaa,rapier}, including approximation algorithms for average completion time.

\subsubsection*{Average response time}
The single machine preemptive case with release times, $1|pmtn,q_i|\sum_i R_i$, is solvable in polynomial time using the shortest remaining processing time (SRPT) rule \cite{baker:schedule}. Without preemption, $1||\sum_i R_i$ is solvable using the shortest processing time (SPT) rule; but, $1|q_i|\sum_i R_i$ is \ignore{$\mathbf{NP}$-hard \cite{lenstra1977sequencing}. Furthermore, Kellerner et al. show that this case is} hard to approximate within a factor of $n^{\frac{1}{2}-\epsilon}$ for all $\epsilon>0$ \cite{kellerer+tw:schedule}. \ignore{In the same paper, the authors give an $O(\sqrt{n})$ approximation algorithm as well.} For two machines or more, $P2|pmtn,q_i|\sum_i R_i$ is $\mathbf{NP}$-hard \cite{du+ly:schedule}. Leonardi and Raz show that SRPT is an $O(\log(\min(\frac{n}{m},P)))$-competitive algorithm for the problem $Pm|pmtn,q_i|R_i$ where $P$ is the ratio between the largest and the smallest job processing times \cite{leonardi+r:schedule}. \ignore{They also establish an $\omega(\log P)$ lower bound on the competitiveness of any randomized online algorithm. Awerbuch et al. give an online non-migratory queue-based algorithm algorithm (for the same problem) with a competitive ratio of $O(\min(\log P, \log n))$ \cite{awerbuch+alr:schedule}. Chekuri, Khanna, and Zhu present a simple $O(\min(\log P, \log\frac{n}{m}))$-competitive non-migratory queue-based algorithm \cite{chekuri+kz:schedule}. Avrahami and Azar obtain the same bound with a new non-migratory algorithm which does not need a queue \cite{Avrahami}.}  From a technical standpoint, a related paper for our work is that of Garg and Kumar, who consider the problem of minimizing total response time on related machines ($Q|pmtn,q_i|\sum_i R_i$) and present an offline $O(\log P)$-approximation algorithm and an online $O(\log^2 P)$-competitive algorithm \cite{Garg2006}. In a later paper, the same authors consider the problem of minimizing total response time on multiple identical machines where each job can be assigned to a specified subset of machines. They give an $O(\log P)$-approximation algorithm as well as an $\Omega(\frac{\log P}{\log\log P})$ lower bound \cite{Garg2007}. The same ideas were used to get an $O(k)$-approximation algorithm for the unrelated case ($R|pmtn,q_i|\sum_i R_i$) when there are $k$ different processing times \cite{Garg2008}. In the same paper, the authors showed an $\Omega(\log^{1-\epsilon}P)$ hardness of approximation for $P|pmtn,q_i|\sum_i R_i$. More recently, Bansal and Kulkarni design an $O(\min(\log^2 n,\log n\log P))$-approximation algorithm for $R|pmtn,q_i|\sum_i R_i$, which provides a basis for our algorithm for average response time~\cite{Bansal2015flowtime}.

Independently, Dinitz and Moseley \cite{dinitz+moseley:schedulingflowsreconfig} have recently studied online scheduling of flows in reconfigurable networks and
provide an $O(1/\eps^2)$-competitive algorithm, assuming that the speed of each machine \ignore{in the algorithm} is $2+\eps$ times that in an optimal solution.  
\iflong
Although their model, generalized to multi-graphs, captures our model, their notion of speed augmentation is subtly different from our notion of port capacity augmentation. 
For instance, while augmenting speed always leads to faster completion of a job being currently scheduled, increasing the port capacity allows more unit flows to be scheduled in the same round, but each unit flow still takes the full round to complete.  This entails that in our model, port capacities are always integers, while in their model, speeds can take real values.
As a consequence, their result implies an $O(1/c^2)$-competitive algorithm for average response time in our model, assuming a $(2 + c)$ factor blowup in port capacity, for any given positive integer $c$.  In contrast, our result for average response time requires a $(1 + c)$-factor blowup, for any given positive integer $c$, but only holds for the offline model and has a logarithmic approximation ratio.
\fi
\ignore{With weights, the problem becomes considerably more difficult. As a matter of fact, $1|pmtn,q_i|\sum_i \omega_i R_i$ is $\mathbf{NP}$-complete \cite{Lenstra1977343} and $P|pmtn,q_i|\sum_i \omega_i R_i$ is $\mathbf{APX}$-hard \cite{chekuri+k:schedule}. Chekuri and Khanna present a semi-online $O(\log^2 P)$-competitive algorithm and a quasi-polynomial approximation scheme (QPTAS) for $1|pmtn,q_i|\sum_i \omega_i R_i$ \cite{chekuri+k:schedule}. Bansal and Dhamdhere give an $O(\log n + \log P)$-approximation and an $O(\log W)$-competitive algorithm for the same problem, where $W$ is the maximum to minimum ratio of weights~\cite{Bansal+d}. Bansal and Chan show an $\omega(1)$ lower bound on the competitive ratio of any deterministic online algorithm for $1|pmtn,q_i|\sum_i \omega_i R_i$ \cite{bansal+chan}. In a very recent breakthrough, Feige et. al, using results of Batra et. al, presented a polynomial-time algorithm with a $O(1)$ approximation ratio for minimizing weighted response time on a single machine~\cite{batra,feige+kl:flow}.}

\subsubsection*{Maximum response time}
The problem of minimizing maximum response time has not been studied extensively. $P|pmtn,q_i|R_{\max}$ is polynomial-time solvable \cite{Lawler1978}. The first-in first-out (FIFO) heuristic is known to be $(3-\frac{2}{m})$-competitive for $Pm|pmtn,q_i|R_{\max}$ and $Pm|q_i|R_{\max}$ \cite{Mastrolilli2003,bender+cm:schedule}. On the other hand, Ambühl and Mastrolilli give a $(2-\frac{1}{m})$-competitive algorithm for $Pm|pmtn,q_i|R_{\max}$ and show that FIFO achieves the best possible competitive ratio on two identical machines when preemption is not allowed \cite{AMBUHL2005597}. \cite{Bansal2015flowtime} gives an $O(\log n)$-approximation algorithm for $R|pmtn,q_i|R_{\max}$ . 

\section{Problem Definitions and Notation} 
\label{sec:problemstatement}
We consider two scheduling problems in which flows arrive in fixed intervals on a non-blocking switch. In this model, we are given a switch $\instance{m,m'} = (P,\flows{})$ where $P$ is a set of $m$ \textit{input} ports and $m'$ \textit{output} ports where each port $p$ has a corresponding \textit{capacity} $\capacity{p}$. $\flows{}$ is a set of flows $e = pq$ with one input port $p$ and one output port $q$.  Each flow $e$ has a corresponding \textit{demand} $\size{e}$ and \textit{release time} $\release{e}$. We assume throughout that for any $e = pq$, $\size{e} \le \mincapacity{e} = \min(c_p,c_q)$.

For an given instance $\instance{m,m'}$, we define a family of functions $\schedule{}:F \times \mathbb{N} \to \{0,1\}$. We say that $\schedule{}$ \textit{schedules} flow $e$ in \textit{round} $t$ if $\schedule{e,t} =1$ (for ease of notation, we use $\schedule{e,t} \equiv \sigma(e,t)$). A function $\schedule{}$ is a \textit{schedule} of $\instance{m,m'}$ if the following conditions are met: every flow $e$, is entirely scheduled across all rounds (i.e. $\sum_t \schedule{e,t} \ge 1$), every flow $e$ is scheduled only in rounds after its release time (i.e. for all $t$, $\schedule{e,t} = 1 \Rightarrow t \ge \release{e}$), and for all ports $p$ the total size of all flows scheduled on port $p$ in a given round is no more than $p$'s capacity (i.e. for all $t$, $\sum_{e:p\in e} \size{e}\schedule{e,t} \le \capacity{p}$). 
For a given flow $e$ and schedule $\schedule{}$, the \textit{response time} $\response{e}$ is the difference in its \textit{completion time} $\completion{e} = 1 + \min \{t : \schedule{e,t} = 1\}$ and its release time, i.e. $\response{e} = \completion{e} - \release{e}$.

The first problem we study in this model is Flow Scheduling to Minimize Average Response Time (\fsart) in which we seek to minimize $\sum_{e \in \flows{}} \completion{e} - \release{e}$. 
The second problem we study in this model is Flow Scheduling to Minimize Maximum Response Time (\fsmrt) in which we seek to minimize $\max_{e \in \flows{}} \{\completion{e} - \response{e}\}$. 
 
\ignore{
We also study a continuous version of the model described above, in which a schedule $\schedule{} : F \times \mathbb{N} \to [0,1]$ assigns some fraction of each job to each round subject to the following conditions: $\sum_{t} \schedule{e,t} \ge 1$ for all flows $e$, $\sum_{e:p \in e} \schedule{e,t} \le \capacity{p}$ for all rounds $t$, and if $\schedule{e,t} > 0$ then $t \in \free{e}$. In the continuous model, we consider both average and maximum response time objectives.
}

Throughout the paper we use $pq$ to denote a flow (directed edge) from input port $p$ to output port $q$. We use $[i]$ to denote the set of positive integers less than or equal to $i$.  An instance with equal numbers of input and output ports is referred to as $\instance{m}$. The main notation is given in the table below.

\smallskip

\ifsigconf
\else 

\begin{center}
\begin{minipage}[t]{42.5mm}%
    \noindent\begin{tabular}{@{}|ccl|}
        \hline
        $\instance{m,m'}$ &:& $m$-in, $m'$-out \\ 
        $P$ &:& all ports  \\
        $\flows{}$ &:& all flows  \\
        $n$ &:& $|F|$\\
        \hline
        $p,q$ &:& port\\
        $\capacity{p}$  &:& $p$'s capacity\\
        $\mincapacity{pq}$ &:& $\min\{\capacity{p}, \capacity{q}\}$\\
        $\flows{p}$ &:& all $e : p \in e$ \\
        \hline
    \end{tabular}
\end{minipage}%
\begin{minipage}[t]{60mm}
    \begin{tabular}{|ccl|}
        \hline
        $e$, $pq$ &:& flow \\
        $\size{e}$ &:& $e$'s demand \\
        $\release{e}$ &:& $e$'s release time  \\
        $\response{e}$ &:& $e$'s response time \\
        $\completion{e}$ &:& $e$'s completion time\\
        \hline
        $t$ &:& round \\
        $\sigma{}$ &:& schedule \\
        $\schedule{e,t} = 1$ &$\Leftrightarrow$& $e$ scheduled at $t$ \\
        \hline
    \end{tabular}
\end{minipage}
\end{center}

\fi

\section{Average Response Time} 
\label{sec:avgresponsetime}
We study Flow Scheduling to Minimize Average Response Time (\fsart), for instances with identical numbers of input and output ports. Specifically, we assume each instance is an $m \times m$ switch $\instance{m}$.

From a complexity viewpoint, \fsart\ generalizes classic scheduling problems.  The special case of \fsart\ with arbitrary demands, unit capacity, and $m = 1$
 is equivalent to preemptive single-machine scheduling with release times, which is strongly $\mathbf{NP}$-hard when the objective is weighted sum of completion times ($1|r_i; pmtn|\sum w_i c_i$). Note that, $1|r_i; pmtn|\sum c_i$ is polynomial-time solvable while the complexity of $1|r_i; \sigma_i = \sigma; pmtn|\sum w_i c_i$ is still open.

For $m > 1$, \fsart\ instances incur coupling issues, even for unit demands. Each flow requires resources at two ports \emph{simultaneously}\iflong, which makes the problem harder in a different way\fi.  In \cite{chromatic-sched}, the authors consider the closely related {\em biprocessor scheduling}\/ problem: there are $m$ identical machines and $n$ unit-sized jobs which require simultaneous use of two pre-specified (dedicated) machines. The objective is to minimize total completion time of jobs. The hardness of this problem is related to the graph that arises from the pre-specified machine pairs (machines correspond to nodes and edges to jobs). It is shown in~\cite{chromatic-sched} that the problem is strongly NP-hard if the graph is cubic, and remains NP-hard if the graph is bipartite and subcubic (i.e. $\forall v : deg(v) \leq 3$), which implies that \fsart\ is NP-hard even for unit demands and unit capacities and identical release times for all flows. While constant-factor approximations~\cite{chromatic-sched,kub-kraw} are known for makespan and average completion time, no results are known for response time metrics. 

\ignore{
On the positive side, \david{is this problem equivalent to ours if the graph is bipartite?}
\begin{itemize}
\item There exists a greedy 2-approximation algorithm for this problem (for any graph).
\item If the bipartite graph is such that the degree of all vertices (except possibly one) is the same, then a polynomial-time algorithm exist.
\item If the graph is a tree, there exists a polynomial time algorithm.
\item If the graph is a star or path, polynomial-time algorithms exist even when the job sizes vary.
\end{itemize}

To summarize, our problem with unit-size flows and unit port capacities is strongly NP-had even when all flows have the same release time and the objective function is unweighted. On the positive side, the 2-approximation algorithm might be useful as a sub-routine.

Finally, we note that if the objective function is makespan $C_{\max}$, the problem is still $NP$-hard (even in the case of unit-sized jobs); however, a 4-approximation algorithm exists \cite{kub-kraw}.}

Section~\ref{sec:linear_prog_approach} presents a linear programming approach based on iterative rounding,\iflong building on prior work on unrelated machines\fi.  Section~\ref{sec:validschedule} uses this approach to establish the main approximation result of this section. 

    
    
    \subsection{A linear-programming approach}
    \label{sec:linear_prog_approach}
    In this section, we investigate linear programming approaches used in the context of machine scheduling and adapt them to our setting. On a conceptual level, our problem is harder than parallel/related/unrelated machine scheduling in the sense that we have to deal with simultaneous use of ports, but is easier in the sense that we do not have to worry about the assignment of flows/jobs to machines as each flow specifies its source and destination ports.

Our starting point is the following linear program similar to the one used by Garg and Kumar~\cite{Garg2006}.

\vspace{-3mm}
\begin{align} 
\label{eq:lp_gk_s}
\text{Minimize}\ &\sum_e \sum_{t\geq \release{e}} \Big(\frac{t-\release{e}}{\size{e}}+\frac{1}{2\mincapacity{e}} \Big)\ b_{et} &\text{subject to}
\\
&\sum_{t\geq \release{e}} b_{et} \geq \size{e} &\forall e
\label{eq:request_const}
\\
&\sum_{e \in \flows{p}} b_{et} \leq \capacity{p} &\forall p, t \label{eq:port_const}
\\
&b_{et} \geq 0 & \forall e, t \label{eq:lp_gk_f}
\end{align}

Informally, the variable $b_{et}$ gives the amount of flow $e$ that is scheduled in round $t$.
Constraint~(\ref{eq:request_const}) ensures that each flow is completed. Constraint~(\ref{eq:port_const}) ensures that no port is overloaded in any round. We can rewrite the objective function as $\sum_e \fresponse{e}$ where 
\begin{equation*}
\fresponse{e} = \sum_{t\geq \release{e}} \Big(\frac{t-\release{e}}{\size{e}}+\frac{1}{2\mincapacity{e}}\Big)\ b_{et}
\label{eq:frac_res_time}
\end{equation*}
is the \textit{fractional response time} of $e$. We show that, for a given instance $\instance{n,m}$ of \fsart, the optimal solution to \program{\ref{eq:lp_gk_s}}{\ref{eq:lp_gk_f}} lower bounds the total response time of any schedule of $\instance{n,m}$.

\begin{lemma}
    For an arbitrary $\instance{n,m}$, let $\schedule{}$ be some (non-integral) schedule of $\instance{n,m}$ and let $b^*$ ($\fresponse{e}^*$) be the optimal solution to \program{\ref{eq:lp_gk_s}} {\ref{eq:lp_gk_f}} corresponding to $\instance{n,m}$. Then $\sum_e \fresponse{e}^* \le \sum_{e} \response{e}$. 
\end{lemma}

\ignore{
We want to show that the value of an optimal solution the the above LP is a lower bound on the total integral response time. Suppose $\schedule{}$ is a valid schedule and let $\response{e}(S)$ denote the integral response time of flow $f$ according to $S$. Let $b_S$ denote the feasible solution to the LP corresponding to $S$.
\begin{lemma}
$\forall f : \Delta_f(b_S) \leq \rho_f(S)$.
\end{lemma}
}

\begin{proof}
Given $\schedule{}$, we construct a solution to \program{\ref{eq:lp_gk_s}}{\ref{eq:lp_gk_f}} by setting $b_{et} \leftarrow (1/\size{e}) \schedule{e,t}$, for all flows $e$ and rounds $t$.
To prove the lemma, we prove the stronger claim that, for any flow $e$, $\fresponse{e} \le \response{e}$.

Suppose that the completion time of flow $e$ in schedule $\schedule{}$ is $\completion{e}$. Then the response time of $e$ is $\response{e} = \completion{e} - \release{e}$. Notice that 
\begin{equation*}
\fresponse{e} = \sum_{t = \release{e}}^{\completion{e}}  \Big(\frac{t-\release{e}}{\size{e}}+\frac{1}{2\mincapacity{e}}\Big)\ b_{et}\ \leq\ \sum_{t = \completion{e}-\size{e}/\mincapacity{e}}^{\completion{e}} \Big(\frac{t-\release{e}}{\size{e}}+\frac{1}{2\mincapacity{e}}\Big)\ \mincapacity{e}.
\end{equation*}
That is, $\fresponse{e}$ is maximized when as much of flow $e$ is scheduled in each round as possible to ensure that $e$ completes in round $\completion{e}$. But,
\begin{align*}
\sum_{t = \completion{e} - \size{e}/\mincapacity{e}}^{\completion{e}} \Big( \frac{t- \release{e}} {\size{e}}+\frac{1}{2\mincapacity{e}} \Big)\ \mincapacity{e} &= \sum_{t = 1}^{\size{e}/\mincapacity{e}} \Big(\frac{\completion{e}-\release{e}-t}{\size{e}} + \frac{1}{2\mincapacity{e}}\Big)\ \mincapacity{e}
\\
&= \completion{e} - \release{e} - \frac{1}{2} \leq \response{e}
\end{align*}
 which completes the proof.
\end{proof}

\iflong
\begin{remark}
\emph{We note that the optimal solution to \program{\ref{eq:lp_gk_s}}{\ref{eq:lp_gk_f}} yields a non-integral schedule which optimizes average response time. 
\ignore{
It is worthwhile to note that the optimal solution to the LP is a valid schedule in the continuous model. Recall that in the continuous model, a schedule consists of flow bandwidth functions respecting port capacities.}
Importantly, the solution already takes care of the resource coupling issue (between ports) for us. Unfortunately, it is not clear what is the gap between the LP's objective function and the true total response time.}
\end{remark}
\fi 
We now consider another linear programming formulation first used by Bansal and Kulkarni \cite{Bansal2015flowtime} for the problem of job scheduling on unrelated machines. The authors use  iterative rounding to get a tentative schedule with low additive overload for any interval of time. We do the same. This linear program and the subsequent ones, used in iterative rounding, are all interval-based. In the initial program, which we denote $LP(0)$, the interval size is 4.  In subsequent relaxations, the interval size can grow. $LP(0)$ is the following program along with constraints (\ref{eq:request_const}) and (\ref{eq:lp_gk_f}). 

\ignore{
\begin{remark}
\emph{In what follows, we make two simplifying assumptions. First, all flows are unit sized. Second, all port capacities are unit. Later we show how to remove these restrictions.}
\end{remark}
}

\vspace{-3mm}
\begin{align}
\label{lp:avg0_start}
\text{Minimize}\ &\sum_e \sum_{t\geq \release{e}} \Big( \frac{t-\release{e}}{\size{e}}+\frac{1}{2}\Big)\ b_{et} &\text{subject to} 
\\
&\sum_{e \in \flows{p}} \, \sum_{t \in (4(a-1),4a]} b_{et} \leq 4\capacity{p} &\forall p, a 
\label{lp:avg0_capacity}
\end{align}

As before, the real variable $b_{et}$ is the amount of flow $e$ scheduled in round $[t,t+1)$. Constraint (\ref{lp:avg0_capacity}) ensures that the total sum of flows scheduled on a given port $p$ in any four consecutive rounds is no more than four times the capacity of $p$. Clearly, this new LP is a relaxation of the previous one; consequently, the value of an optimal solution to this LP is a lower bound to the response time for any integral schedule.
Following \cite{Bansal2015flowtime}, we use an iterative rounding scheme to get the following result.

\begin{lemma}\label{lem:lpsolution}
The exists a solution $b^{*} = \{b^{*}_{et}\}_{e,t}$ satisfying the following properties
\begin{enumerate}
    \item For each flow $e$, there is exactly one round $t$ for which $b^{*}_{et}=\size{e}$.
    \item The cost of $b^{*}$ is at most that of an optimal solution to the LP.
    \item For any port $p$ and any time interval $[t_1,t_2]$,
    \begin{equation*}
    \sum_{e \in \flows{p}}\ \sum_{t \in [t_1,t_2]} b^{*}_{et}\ \leq\ \capacity{p}(t_2-t_1) + O(\capacity{p} \log n). 
    \end{equation*}
\end{enumerate}
\end{lemma}

\begin{remark}
    \emph{We can regard a solution satisfying the properties in the lemma as a sequence of bipartite graphs $\{G_t\}_t$. Then, for any given (time) interval $[a,b]$, the degree of any vertex $p$ in the ``combined'' graph $\cup_{t\in [a,b]} G_t$ is at most $(b-a)\capacity{p} + O(\capacity{p}\log n)$. 
    In Section \ref{sec:validschedule}, we convert this sequence to a sequence of matchings.
    }
\label{rem:bipartitedegree}
\end{remark}

\ignore{\emph{Consider a solution $b^{*}$ satisfying the three properties in the lemma. This solution can be regarded as a sequence of bipartite graphs $\{G_t\}_t$ such that for any given (time) interval $[a,b]$, the degree of any vertex $p$ in the ``combined'' graph $\cup_{t\in [a,b]} G_t$ is at most $\capacity{p}$ times the length of the interval, plus a factor of $\capacity{p}\log n$. In the next section, we will examine how to convert this sequence to a sequence of matchings.}}
    
        \subsubsection*{Iterative rounding}
        \label{sec:avgiterativerounding}
        
To establish Lemma \ref{lem:lpsolution}, we iteratively relax variable assignments with a sequence of linear programs which we denote by $LP(\ell)$ for $\ell = 0, 1, ...$. Recall that $LP(0)$ is the initial linear program above.  
We denote the set of flows that appear in $LP(\ell)$ by $F(\ell)$ and an optimal solution to $LP(\ell)$ by $b^{\ell} = \{b^{\ell}_{et}\}_{e,t}$.
Let $E(\ell)$ be the set of variables in $LP(\ell)$ with non-zero assignments. Let $A(\ell)$ be the set of flows $e$ such that, for all $t$, $b_{et}^{\ell}$ is integral. Let $P(\ell)$ be the set of tight capacity constraints (\ref{lp:avgl_capacity}) in $LP(\ell)$ given $b^{\ell}$.
Let $\size{\max} = \max_e\{\size{e}\}$.
See Figure~\ref{fig:iterativeroundingoverview} for a high level overview.

\ifsigconf
\else
\begin{figure}
    \centering
    \includegraphics[]{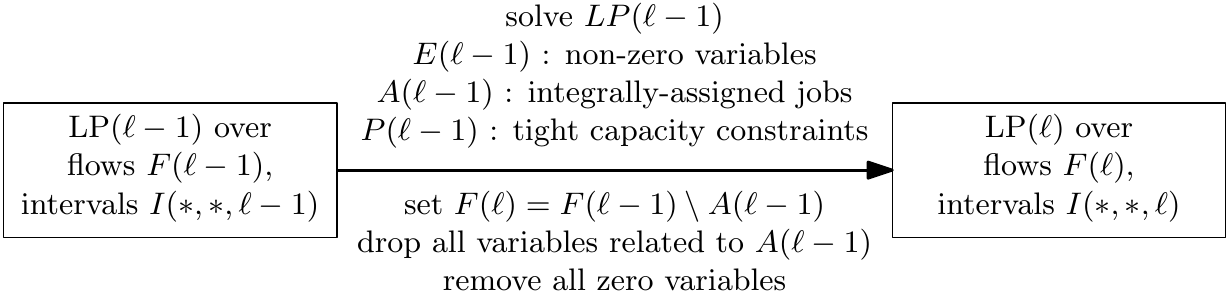}
    \caption{The $\ell$-th iteration of the rounding scheme, $\ell \geq 1$, starts by solving $LP(\ell-1)$ and ends by defining $LP(\ell)$.}
    \label{fig:iterativeroundingoverview}
\end{figure}
\fi

\ignore{
\begin{figure}
\begin{center}
\begin{tikzpicture}[scale=0.85]
\tikzstyle{every node}=[font=\small]

\node at (0,0) (A) [fill=blue!25] {\begin{tabular}{c} $LP(\ell-1)$ over \\ flows $F(\ell-1)$, \\ intervals $I(*,*,\ell-1)$ \end{tabular}};

\node at (12,0) (B) [fill=blue!25] {\begin{tabular}{c} $LP(\ell)$ over \\ flows $F(\ell)$, \\ intervals $I(*,*,\ell)$ \end{tabular}};

\draw[-{Latex[width=3mm, red]}, thick] (A) edge node[anchor=south] {\begin{tabular}{c} solve $LP(\ell-1)$ \\ $E(\ell-1) :$ non-zero variables \\ $A(\ell) :$ integrally-assigned jobs \\ $P(\ell-1) :$ tight port constraints (\ref{eq:port_const})
\end{tabular}} (B);

\draw[-{Latex[width=3mm, red]}, thick] (A) edge node[anchor=north] {\begin{tabular}{c} set $F(\ell) = F(\ell-1) \backslash A(\ell)$ \\ drop all variables related to $A(\ell)$ \\ remove all zero variables \end{tabular}} (B);

\end{tikzpicture}
\end{center}
\caption{The $\ell$-th iteration of the rounding scheme, $\ell \geq 1$, starts by solving $LP(\ell-1)$ and ends by defining $LP(\ell)$ to be solved in the next iteration.}\label{fig:iteration}
\end{figure}
}

In each iteration $\ell \geq 1$, we construct $LP(\ell)$ as follows. 
\begin{mylowitemize}
\item Initialize $F(\ell)=F(\ell-1)$.
\item Find an optimal solution $\{b^{\ell-1}_{et}\}_{e,t}$ to $LP(\ell-1)$.
\item Eliminate zero variables. In other words the variables $b_{et}$ in $LP(\ell)$ are only defined for variables in $E(\ell-1)$, the support of $b^{\ell-1}$.
\item Fix integral assignments. For all $e \in A(\ell-1)$,  assign $e$ to those rounds $t$ such that $b_{et}^{\ell-1} > 0$ (i.e. set $b^*_{et} \leftarrow \size{e}$) and drop all variables $b_{et}$ in $LP(\ell)$. We also update $F(\ell) = F(\ell) \backslash \{e\}$.  
\ignore{
Let $A(\ell)$ denote the set of all flows which get integrally assigned in the $\ell$-the iteration; that is, the ones which get integrally assigned by solving $LP(\ell-1)$. We define $A(0) = \emptyset$.}
\item Define intervals for the current iteration as follows. Fix a port $p$ and consider the flows in $F(\ell) \cup \flows{p}$. Sort all the variables in $\{b^{\ell-1}_{et} \in E(\ell-l) : e \in F(\ell) \cup F_p \}$ in increasing order of $t$, breaking ties lexicographically. Next, iteratively partition $b^{\ell-1}_{et}$ variables into groups $I(p,1,\ell), I(p,2,\ell),\ldots$ as follows. To construct group $I(p,a,\ell)$, start from the earliest non-grouped variable and greedily group consecutive $b^{\ell-1}_{et}$ variables until their sum first exceeds $4\capacity{p}$. 
The size of the interval $I = I(p,a,\ell)$ is 
\begin{equation*} 
\textrm{Size}(I) = \sum_{b_{et}\in I} b^{\ell-1}_{et}.
\end{equation*}
Note that $\textrm{Size}(I(p,a,\ell)) \in [4\capacity{p}, 5\capacity{p})$. The time duration of $I$ can be much larger than its size.
On the other hand, for $\ell=0$, all intervals are of size 4 as evident in the initial LP.
\end{mylowitemize}

For $\ell \ge 1$, $LP(\ell)$ is given by objective (\ref{lp:avg0_start}) subject to constraints (\ref{eq:request_const}),  (\ref{eq:lp_gk_f}), and the following constraint.
\begin{align}
\sum_{e \in \flows{p} \cap F(\ell)} \, &\sum_{b_{et} \in I(p,a,\ell)} b_{et} \leq \text{Size}(I(p,a,\ell)) \cdot \capacity{p} &\forall p, a
\label{lp:avgl_capacity}
\end{align}

Since $LP(\ell)$ is a relaxation of $LP(\ell-1)$, the second requirement of Lemma~\ref{lem:lpsolution} is satisfied. Also, by construction of $LP(\ell)$, the sequence of iterations results in an integral assignment of all flows and so the first requirement of Lemma~\ref{lem:lpsolution} is satisfied. 
It remains to bound the number of iterations and calculate the backlog.

Recall that $F(\ell)$ is the set of flows $e$ such that variables $b_{et}$ appear in $LP(\ell)$. 
Note that, for $\ell > 0$, these are the non-zero variables which correspond to non-integrally-assigned jobs after solving $LP(\ell-1)$.

\begin{lemma}
For all $\ell \geq 1$, $|F(\ell)| \leq |F(\ell-1)| / 2$.
\label{lem:num_iterations}
\end{lemma}

\begin{proof}
Consider a linearly independent set of tight constraints in $LP(\ell-1)$. Since a tight non-negativity constraint (\ref{eq:lp_gk_f}) results in a zero variable, the number of non-zero variables, $E(\ell-1) \subseteq \{b^{\ell-1}_{f,t}\}_{f,t}$, is at most the number of tight flow constraints (\ref{eq:request_const}) plus the number of tight capacity constraints (\ref{lp:avgl_capacity}). That is
\begin{equation}
|E(\ell-1)| \leq |F(\ell-1)| + |P(\ell-1)|. \label{eq:supp1}
\end{equation}
since $|F(\ell-1)|$ is the number of flow constraints.

Now, each flow which is not integrally assigned by $b^{\ell-1}$ (i.e. not in $A(\ell-1)$) contributes at least two to $|E(\ell-1)|$. Thus,
\begin{align}
|E(\ell-1)| &\geq |A(\ell-1)| + 2(|F(\ell-1)| - |A(\ell-1)|) \notag
\\
&= |F(\ell-1)| + |F(\ell)|. \label{eq:supp2}
\end{align}
The equality holds since $F(\ell) = F(\ell-1) \backslash A(\ell)$ by construction.
Inequalities (\ref{eq:supp1}) and (\ref{eq:supp2}) together imply
\iflong
\begin{equation*}
|F(\ell)| \leq |P(\ell-1)|.
\end{equation*}
\fi

Next, we show that $|P(\ell-1)| \leq |F(\ell-1)| / 2$ which completes the proof. This is accomplished by a simple combinatorial argument. Let's give $2$ tokens to every flow in $F(\ell-1)$. Now, each flow $e \in F(\ell-1)$, gives a portion equal to $b^{\ell-1}_{et}/\size{e}$ of its tokens to the interval that contains $b_{et}$. This token distribution is valid since
\begin{equation*}
\sum_p \sum_t b^{\ell-1}_{et} = 2\sum_t b^{\ell-1}_{et} = 2\size{e},
\end{equation*}
where we have used the fact that each flow appears in exactly two port constraints. At the same time, each tight capacity constraint for port $p$ receives at least $4$ tokens since interval sizes are $\ge 4\capacity{p}$ by definition and $\size{e} \le \capacity{p
}$ by assumption. 
Now, as each job distributes exactly 2 tokens and each tight port constraint receives at least 4, we conclude that $|P(\ell-1)| \leq |F(\ell-1)| / 2$.
\end{proof}

Lemma~\ref{lem:num_iterations} shows that the number of iterations needed before arriving at an integral solution is no more than $O(\log n)$.
What remains is to bound amount of extra load that any interval has taken on. Recall that $A(\ell)$ denotes the set of flows which are integrally assigned by the optimal solution $b^{\ell}$ to $LP(\ell)$. Let $A(\ell, p, t_1, t_2) \subseteq A(\ell)$ be the set of flows which are integrally assigned to port $p$ in the interval $[t_1, t_2]$ by the optimal solution $b^{\ell}$ of $LP(\ell)$. Furthermore, we define
\begin{equation*}
\mathrm{Vol}(p,\ell,t_1,t_2) = \sum_{e \in \flows{p} \cap F(\ell)} \sum_{t\in [t_1,t_2]} b^{\ell}_{et} + \sum_{\ell' \leq \ell} |A(\ell',p,t_1,t_2)|,
\end{equation*}
which is the total size of flows assigned to port $p$ in the interval $[t_1,t_2]$\iflong, either integrally or fractionally,\fi by $b^0,\ldots,b^{\ell}$.
The following lemma states that the amount of extra load taken on any port in any interval is no more than a constant additive over the load in the previous iteration.

\begin{lemma}\label{lem:backlog}
For any $[t_1,t_2]$, any port $p$, and any round $\ell \geq 1$,
\begin{equation} \label{eq:vollem}
\mathrm{Vol}(p,\ell,t_1,t_2) \leq \mathrm{Vol}(p,\ell-1,t_1,t_2) + 10\capacity{p}.
\end{equation}
\end{lemma}

\begin{proof}
Fix an interval $[t_1,t_2]$ and a port $p$. In each iteration $\ell$, the ``extra'' load in this interval can be introduced only if two intervals overlap with the boundaries of $[t_1,t_2]$.

Consider a maximal set of contiguous intervals $I(\ell,a,p)$, $I(\ell,a+1,p)$, ..., $I(\ell,a+w,p)$ that contain $[t_1,t_2]$. Note that $a$ is the smallest index such that $I(\ell,a,p)$ contains some $b^{\ell}_{et}$ with $t\in[t_1,t_2]$. Similarly, $w$ is the largest number such that $I(\ell,a+w,p)$ contains some $b^{\ell}_{et}$ with $t\in[t_1,t_2]$. Since each interval is of size smaller than $5\capacity{p}$, 
\begin{equation}
\sum_{b_{et}\in I(\ell,a,p)} b^{\ell}_{et} + \sum_{b_{et}\in I(\ell,a+w,p)} b^{\ell}_{et} < 10\capacity{p}.
\end{equation}
Moreover,
\begin{align*}
    \sum_{x=a+1}^{a+w-1} \sum_{b_{et}\in I(\ell,x,p)} b^{\ell}_{et} &\leq \sum_{x=a+1}^{a+w-1} \mathrm{Size}(I(\ell,x,p))\\
    = \sum_{x=a+1}^{a+w-1} \sum_{b_{et}\in I(\ell,x,p)} b^{\ell-1}_{et}
    &\leq \sum_{e \in \flows{p} \cap F(\ell)} \sum_{t\in [t_1,t_2]} b^{\ell-1}_{et},
\end{align*}
where the first inequality follows from the port capacity constraints (\ref{eq:port_const}), and the second equality follows from the definition of $\textrm{Size}(*)$. Consequently, we have that
\begin{align*}
\sum_{e \in \flows{p} \cap F(\ell)} \sum_{t\in [t_1,t_2]} b^{\ell}_{et} &\leq \sum_{x=a}^{a+w} \ \sum_{b_{et}\in I(\ell,x,p)} b^{\ell}_{et} \\
& < 10\capacity{p} + \sum_{e \in \flows{p} \cap F(\ell)} \ \sum_{t\in [t_1,t_2]} b^{\ell-1}_{ft}
\end{align*}
\[ \hspace{1.5cm}\leq 10\capacity{p} - |A(\ell-1, t_1, t_2, p)| + \sum_{e \in \flows{p} \cap F(\ell-1)} \ \sum_{t\in [t_1,t_2]} b^{\ell-1}_{et}, \]
where the last step uses the fact that $F(\ell) = F(\ell-1) \backslash A(\ell)$.

The LHS of (\ref{eq:vollem}) equals
\begin{align*}
\sum_{f \in F_p \cap F(\ell)} \sum_{t\in [t_1,t_2]} b^{\ell}_{ft} + &|A(\ell-1, t_1, t_2, p)| \leq\\
&\sum_{f \in F_p \cap F(\ell-1)} \sum_{t\in [t_1,t_2]} b^{\ell-1}_{ft} + 10\capacity{p},
\end{align*}
which equals the RHS of (\ref{eq:vollem}). 
\end{proof}
We now establish a bound on the total ``extra'' load in any interval for the final assignment. Recall that $b^*$ is the final, integral assignment derived from the iterative procedure above. 

\begin{lemma} \label{lem:error_bound}
For any interval $[t_1,t_2]$ and port $p$, 
\[\sum_{e \in \flows{p}}\ \sum_{t \in [t_1,t_2]} b^{*}_{et}\ \leq\ \capacity{p}(t_2-t_1) + 10\capacity{p}\log n.\]
\end{lemma}

\begin{proof}
    We fix the interval $[t_1,t_2]$ and port $p$. By construction of $b^*$, we need only to show that, for all $\ell$ 
    \begin{equation}
        \mathrm{Vol}(p,\ell,t_1,t_2) \le \capacity{p}(t_1-t_2) + 10(\ell+1) \capacity{p}.
        \label{eq:overload}
    \end{equation}
    We prove inequality~\ref{eq:overload} by induction on $\ell$.
    For $\ell = 0$, we have 
    \begin{align*}
        \mathrm{Vol}(p,0,t_1,t_2) = \sum_e \sum_{t \in [t_1,t_2]} b^0_{et} \le \capacity{p}(t_1 - t_2) + 4 \capacity{p}
    \end{align*}
    by Constraint~(\ref{lp:avg0_capacity}).
    So
    \begin{align*}
        \mathrm{Vol}(p,\ell+1,t_1,t_2) &\le \mathrm{Vol}(p,\ell,t_1,t_2) + 10\capacity{p} \\
        &\le \capacity{p}(t_1 - t_2) + 10(\ell + 2)\capacity{p}
    \end{align*}
\iflong
where the first inequality follows by Lemma~\ref{lem:backlog} and the second by induction.
\fi 
\end{proof}

We now have all the necessary ingredients to prove Lemma \ref{lem:lpsolution}.
\begin{proof}[Proof of Lemma \ref{lem:lpsolution}]
In the final solution $\{b^{*}_{et}\}_{e,t}$, all flows are integrally assigned. Furthermore, the cost of the final solution is at most that of an optimal solution to the initial linear program (since each iteration, we are relaxing the previous linear program). Finally, by Lemma \ref{lem:error_bound}, for any time interval $[t_1,t_2]$ and port $p$, the total volume of assigned flows is at most $\capacity{p}(t_2-t_1) + O(\capacity{p}\log n)$. 
\end{proof}
        
    \subsection{Getting a valid schedule}
    \label{sec:validschedule}
    What we obtain from Lemma \ref{lem:lpsolution} is, unfortunately, not a valid schedule but what could be called a \emph{pseudo-schedule}; as noted in Remark~\ref{rem:bipartitedegree}, the total amount of flow passing through a port $p$ during a time interval $I$ could as much as $c_p O(\log n)$ more than $c_p |I|$, as allowed by the capacity of the port.  In this section we show that we can convert the pseudo-schedule given by Lemma \ref{lem:lpsolution} into a valid schedule using {\em resource augmentation}, i.e., assuming the algorithm is allowed more port capacity than the optimal schedule.  It is immediate from Lemma~\ref{lem:lpsolution} that if we augment the capacity of every port by a factor of $1 + O(\log n)$, then we obtain a valid resource-augmented schedule with optimal average response time.  In the following, we show that we can achieve logarithmic-approximate average response time with a small constant blowup in port capacity, for the case of unit demand flows (and arbitrary port capacities).

\begin{theorem}\label{thm:speed_aug}
For any positive integer $c$, there exists a polynomial-time algorithm that, given a set of $n$ unit flows over a switch, computes a $(1 + \frac{O(\log n)}{c})$-approximation for average response time unit-size flows, while incurring a blowup in capacity by a factor of $(1 + c)$.
\end{theorem}
\begin{proof}
Given a set $F$ of flows over a switch, by Lemma \ref{lem:lpsolution}, there exists a pseudo-schedule which assigns flows to time slots such that the total response time is at most the cost of an optimal solution to the initial linear program and for any given time interval $[t_1,t_2]$, and for any port $p$, the total volume of flows assigned to $p$ during the interval is at most $c_p(t_2-t_1) + O(c_p\log n)$.

\iflong
We first prove the desired claim for unit capacities.
\fi
The  pseudo-schedule can be regarded as a sequence $\{G_t\}_t$ of bipartite $m\times m$ graphs such that in any given interval $[t_1,t_2]$, the degree of each vertex in the combined graph $\cup_{t_1}^{t_2} G_t$ is at most $(t_2-t_1)+c'\log n$ for some $c'>0$. Next, we convert this sequence $\{G_t\}_t$ into a sequence of bipartite matchings $\{M_t\}_t$. To this end, we divide the timeline into consecutive intervals $I_1, I_2, ...$, each of size $h=\lceil\frac{c\log n}{c}
\rceil$. Now, starting from the beginning, we schedule flows in each interval before going to the next one. Consider an interval $I_j$, the degree of each vertex in the combined graph $G_{I_j}$ is at most $d=\lceil c'(1+\frac{1}{c})\log n\rceil$. Applying the Birkhoff-von Neumann Theorem \cite{birkhoff}, $G_{I_j}$ can be decomposed into at most $d$ matchings in polynomial time. By increasing the capacity (bandwidth) of each port to $1+c$, we can execute $d$ matchings in the next available spots (with respect to release times) in at most $h$ time steps. Since each flow is delayed by at most $h+d=\frac{O(\log n)}{c}$ steps, the total response time of this schedule is no more than
\begin{equation}
OPT\ +\ n\times\frac{O(\log n)}{c}\ \leq\ OPT\times(1 + \frac{O(\log n)}{c}),
\end{equation}
where the inequality follows from the fact that the number of flows is lower bound on the total response time.
\iflong
We now show that the above algorithm and argument can be extended to general capacities, using the notion of {\em $b$-matchings}\footnote{A $b$-matching of a bipartite graph, for a given function $b$ from the graph's vertex set to nonnegative integers, is a subgraph in which the degree of each vertex $v$ is at most $b(v)$ (e.g., see~\cite{gerards:matching}).} and a standard transformation between $b$-matchings and matchings~\cite{gerards:matching}.  In the general case, the 
pseudo-schedule can be regarded as a sequence $\{G_t\}_t$ of bipartite graphs such that in any given interval $[t_1,t_2]$, the degree of port $p$ in the combined graph $\cup_{t_1}^{t_2} G_t$ is at most $c_p(t_2-t_1)+c_p \cdot c'\log n$ for some $c'>0$.  Similar to the unit capacity case, we convert this sequence $\{G_t\}_t$ into a sequence $\{M_t\}_t$ of bipartite $b$-matchings, where the function $b$ corresponds the port capacities. To this end, we divide the timeline into consecutive intervals $I_1, I_2, ...$, each of size $h=\lceil\frac{c'\log n}{c}
\rceil$. Now, starting from the beginning, we schedule flows in each interval before going to the next one. For each interval $I_j$, we construct a bipartite graph $B_j$ as follows.  We replicate each port $p$ $c_p$ times, and process the edges of $G_{I_j}$ in sequence: for edge $(p,q)$, we add an edge to $B_j$ between a copy of $p$ and a copy of $q$, each of which is chosen in a round-robin manner among the copies of $p$ and $q$, respectively.  This ensures that the degree of any vertex in $B_j$ is at most $d=\lceil c'(1+\frac{1}{c})\log n\rceil$. Now, applying the Birkhoff-von Neumann Theorem \cite{birkhoff}, $B_j$ can be decomposed into at most $d$ matchings in polynomial time. By increasing the capacity (bandwidth) of each port replica in each $B_j$ to $1+c$, and hence increasing the capacity of each port $p$ by a factor of $1+c$, we can execute $d$ matchings in the next available spots (with respect to release times) in at most $h$ time steps, with a $1 + O(\log n)/c$ increase in average response time. 
\fi
\end{proof}

\ignore{
\begin{remark}
Lemma \ref{lem:lpsolution} can be extended to flows of arbitrary size (with or without port capacities). Furthermore, Theorem \ref{thm:speed_aug} can also be extended to the case of arbitrary flow size and unit port capacites resulting in a $(1+\epsilon)$-speed $(1+O(P\log n)/\epsilon)$-approximation algorithm where $P$ is the maximum flow size.
\end{remark}}

\section{Maximum Response Time} 
\label{sec:maxresponsetime}
In this section, we consider the problem of Flow Scheduling to Minimize Maximum Response Time (\fsmrt). More formally, for a given instance $\instance{m,m'}$ of \fsmrt, our goal is to find the minimum $\response{}$ such that there exists a schedule of $\instance{m,m'}$ with maximum response time $\response{}$.
Section~\ref{sec:hardness} establishes that solving \fsmrt\ is \textbf{NP}-hard. Section~\ref{sec:approxmaxresponse} provides a tight approximation to \fsmrt\ via a linear programming relaxation and rounding of a more general problem. 

\ignore{
Section~\ref{sec:onlinemaxresponse} shows how to leverage this approximation to derive an approximation for maximum response time in an online setting. 

To find the minimum response time $\rho$ for a given instance $G$, we binary search over all possible values of $\rho$ and test each for feasibility. That is, for each value of $\rho$ we check if there is some solution to $G$ in which all edges are scheduled no more than $\rho$ rounds after their release. 

More formally, we define Response Time-Constrained Flow Scheduling (RTC) as follows. We are given a switch $S_{m,m'} = (V,E)$ where $V$ is a set of $n$ left ports and $m$ right ports and $E$ is a set of edges $uv$ where $u$ is a left port and $v$ a right port. Each port $v \in V$ has a corresponding capacity $c_v$. We are also given a parameter $\rho$. Each edge $e \in E$ has a corresponding release time $r_e$ and a set of active rounds $\free{e} = \{r_e, \ldots, r_e + \rho\}$. The goal is to schedule all edges $e \in E$ in some round $t_e \in \free{e}$ such that the set of edges scheduled in a single round adjacent to a given port $v$ is no more than $c_v$. We then solve a single instance of Flow Scheduling to Minimize Maximum Response Time by performing binary search to find the optimal response time parameter $\rho$. This reduces Flow Scheduling to Minimize Maximum Response Time to $O(\log n)$ instance of Response Time-Constrained Flow Scheduling.
}

    \subsection{Maximum Response Time Hardness} 
    \label{sec:hardness}
    We establish the hardness of approximation for \fsmrt\, motivating our approximations in Section~\ref{sec:approxmaxresponse}. 

\begin{theorem}
    There is no polynomial time algorithm that solves Flow Scheduling to Minimize Maximum Response Time to within a factor of $4/3$ of optimal, assuming $P \ne NP$.
    \label{thm:np_hardness}
\end{theorem}

Our proof of Theorem \ref{thm:np_hardness} is via a reduction from the Restricted Time-table (RTT) problem, which is shown to be NP-hard in \cite{timetable}. 
\iflong
We redefine RTT here for completeness.
\begin{definition}[Restricted Timetable (RTT) problem]
    Given the following data:
    \begin{enumerate}[label=\roman*.]
        \item $H = \{1,2,3\}$ 
        \item a collection $\{ T_1, \ldots, T_m\}$ with $T_i \subseteq H$ and $|T_i| \ge 2$
        \item a function $g: [m] \to 2^{[m']}$ such that $|g(i)| = |T_i|$
    \end{enumerate}
    determine if there is a function $f:[m] \times [m'] \times H \to \{0,1\}$ such that
    \begin{enumerate}[label=(\roman*), resume]
        \item if $f(i,j,h) = 1$ then $j \in g(i)$
        \label{RTT:teachercanteach}
        \item $j \in g(i)$ iff $\sum_{h \in H} f(i,j,h) \ge 1$ for all $i \in [m]$ and $j \in [m']$
        \label{RTT:classistaught}
        \item $\sum_{i\in [m]} f(i,j,h) \le 1$ for all $j \in [m']$ and $h \in H$
        \label{RTT:oneteacherperclass}
        \item $\sum_{j \in [m']} f(i,j,h) \le 1$ for all $i \in [m]$ and $h \in H$
        \label{RTT:oneclassperteacher}
    \end{enumerate}
\end{definition}
\begin{proof}[Proof of Theorem \ref{thm:np_hardness}]
    We reduce RTT to the feasibility version of \fsmrt\ in which we are given a switch $\instance{m,m'}$ and a response time $\response{}$, and our goal is to check whether or not there exists a schedule with maximum response time at most $\response{}$. Let $I$ be an arbitrary instance of the RTT problem consisting of $H$ ,$\{T_1, \ldots, T_m\}$, and $g:[m] \to 2^{[m']}$ . We reduce $I$ to an instance of \fsmrt\ $\instance{m,m'} = (P,\flows{})$ and $\response{} = 3$. 
    In $\instance{m,m'}$, there are $m$ input ports $p_i$, $i \in [m]$, and $m'$ output ports $q_j$, $j \in [m']$. All ports $p$ have capacity $\capacity{p} = 1$. We construct the set $\flows{}$ according to the following steps (in order).
    \begin{enumerate}
        \item For all $i \in [m]$ and $j \in [m']$, if $j \in g(i)$ then we include an flow from input port $p_i$ to output port $q_j$. 
        \item For each input port $p_i$, we take the minimum $h \in T_i$ and release all flows adjacent to $p_i$ in round $h$.
        \item For all $j \in [m']$, we create three new input ports $w_j, y_j, z_j$. We include the flows $q_jw_j$, $q_jy_j$, and $q_jz_j$ and release these flows in round 4.
        \label{step:setmin}
        \item For all $i$ such that $T_i = \{1,3\}$, we create a new output port $q^*_i$ and three new input ports $w_i, y_i, z_i$. We include an flow $p_i q^*_i$ and release it in round 2. We also include flows $q^*_ip_i$, $q^*_iw_i$, and $q^*_iy_i$ and release them in round 3. 
        \label{step:occupy2}
        \item For all $i$ such that $T_i = \{1,2\}$, we create a new output port $q^*_i$ and three new input ports $w_i, y_i, z_i$. We include an flow $p_i q^*_i$ and release it in round 3. We also include flows $q^*_iv_i$, $q^*_iw_i$, and $q^*_iy_i$ and release them in round 4. 
        \label{step:occupy3}
    \end{enumerate}
    See Figure~\ref{fig:reduction} for a depiction of step \ref{step:occupy2}.
    
    \ifsigconf
    \begin{figure}
        \centering
        \includegraphics[scale=.8]{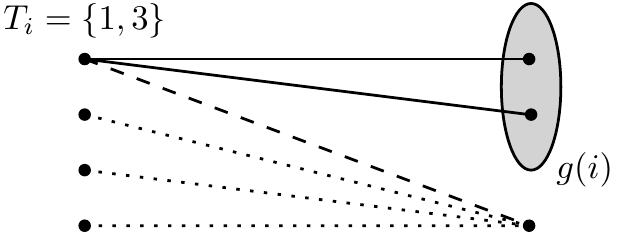}
        \caption{Flows that connect input port $p_i$ with all output ports corresponding to $g(i)$ are shown as solid lines. In this case, the solid flows are released in round 1. In order to prevent any solid flows being scheduled in round 2, we include the dashed flow to be released in round 2 and the three dotted flows to be released in round 3. In order for all three dotted flows to be scheduled by round 5, the dashed flow must be scheduled as soon as it is released.}
        \label{fig:reduction}
    \end{figure}
    \else
    \begin{figure}
        \centering
        \includegraphics[]{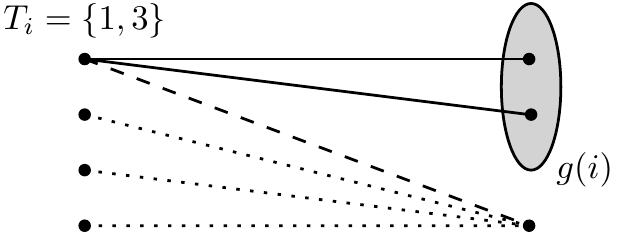}
        \caption{Flows that connect input port $p_i$ with all output ports corresponding to $g(i)$ are shown as solid lines. In this case, the solid flows are released in round 1. In order to prevent any solid flows being scheduled in round 2, we include the dashed flow to be released in round 2 and the three dotted flows to be released in round 3. In order for all three dotted flows to be scheduled by round 5, the dashed flow must be scheduled as soon as it is released.}
        \label{fig:reduction}
    \end{figure}
    \fi
    
    For the remainder of the argument, we refer to the set of ports added in step $i$ as $U_i$ for $i = 3,4,5$, with $U = \bigcup U_i$. For a given function $f : [m] \times [m'] \times H \to \{0,1\}$ we construct a schedule $\schedule{}$ as follows. For all flows $e = p_i q_j \in E$ with $p_i, q_j \in V \setminus U$, we have $\schedule{}$ schedule $e$ in round $h$ if and only if $f(i,j,h) = 1$. 
    For each $j \in [m']$, there are three flows $q_jv ,q_jw$, and $q_jy$ such that $q_j \in V$ and $v,w,y \in U_3$, which $\schedule{}$ schedules in rounds 4, 5, and 6, respectively.
    For all $i$ such that $T_i = \{1,3\}$, there are four flows $p_iu, uw, uy, uz$ with ports $u,v,w,y \in U_4$. $\schedule{}$ schedules $p_iu$ in round 2 and schedules $uw, uy$, and $uz$ in rounds 3, 4, and 5, respectively.
    For all $i$ such that $T_i = \{1,2\}$, there are four flows $p_iq, qw, qy$, and $qz$ with ports $q,v,w,y \in U_5$. $\schedule{}$ schedules $p_iq$ in round 3 and schedules $qw, qy$, and $qz$ in rounds 4, 5, and 6, respectively.
    
    Suppose $f$ satisfies conditions \ref{RTT:teachercanteach}, \ref{RTT:classistaught}, \ref{RTT:oneteacherperclass}, and \ref{RTT:oneclassperteacher}. We show that $\schedule{}$ is a schedule of $\instance{m,m'}$ with maximum response time $\response{}$. By construction of $\schedule{}$, we have that all flows with a port in $U$ are scheduled within three rounds of their release, so we need only to show that all flows $pq$ with $p,q \in P \setminus U$ are scheduled and that there is at most one scheduled flow adjacent to every port in every round. The first follows from condition \ref{RTT:classistaught} of $f$. Suppose there is port that has two adjacent, scheduled flows in one round. By conditions \ref{RTT:oneteacherperclass} and \ref{RTT:oneclassperteacher} of $f$, one of these flows $pq$ must have one left endpoint $p_i \in P \setminus U$ and the other $q \in U$. However, all such flows are scheduled in rounds $h \not\in T_i$, violating condition \ref{RTT:teachercanteach}.
    
    Suppose $\schedule{}$ is a schedule of $\instance{m,m'}$ with maximum response time $\response{}$. We show that $f$ satisfies conditions \ref{RTT:teachercanteach}, \ref{RTT:classistaught}, \ref{RTT:oneteacherperclass}, and \ref{RTT:oneclassperteacher}. 
    By definition of a schedule, $f$ satisfies conditions \ref{RTT:classistaught}, \ref{RTT:oneteacherperclass} and \ref{RTT:oneclassperteacher}. 
    By construction, all flows with an endpoint in $U_3$ must be scheduled in rounds 4, 5, and 6, so all flows $pq$ with $p,q \in P \setminus U$ must be scheduled in rounds 1, 2, or 3. Consider a flow with one endpoint $p_i \in P \setminus U$ and the other $q \in U_4$. Then $q$ has three additional flows that must be scheduled in rounds 3, 4, and 5 in order to schedule all its flows within the response time. So, the flow $p_iq$ must be scheduled in round 2. This entails that all flows $pq$ with $p,q \in V \setminus U$ must be scheduled in rounds 1 and 3. Therefore, $f$ satisfies condition \ref{RTT:teachercanteach}. 
\end{proof}
\fi 

    \subsection{Maximum Response Time Approximation}
    \label{sec:approxmaxresponse}
    In this section, we give an approximation algorithm for Flow Scheduling to Minimize Maximum Response Time (\fsmrt). In fact, the algorithm solves a more general problem which we call \emph{Time-Constrained Flow Scheduling},  for which there is also an easy reduction from \fsmrt.
Time-Constrained Flow Scheduling is identical to \fsmrt\ except that flows do not have corresponding release times. Instead, each flow e, has a corresponding set of (possibly non-contiguous) \textit{active} rounds $\free{e}$ such that $e$ can be scheduled in any round $t \in \free{e}$.
Observe that an instance of \fsmrt\ that is solvable with a maximum response time of $\rho$ can be reduced to an instance of Time-Constrained Scheduling where $R(e) = \{t:\release{e} \le t < \release{e} + \rho\}$ for all flows $e$. Therefore, the approximability of Time-Constrained Scheduling transfers directly to \fsmrt.
\ignore{
We can see that \fsmrt\ reduces to Time-Constrained Scheduling by setting $\free{e} \leftarrow \{t:t\ge\release{e}\}$ for all flows $e$. \david{change to: $\{t:\release{e} \le t \le \release{e} + \rho\}$ for all flows $e$ and a guessed max response time $\rho$.} Therefore, the approximation result for Time-Constrained Scheduling transfers directly to \fsmrt.}

\ignore{
In Time-Constrained Flow Scheduling, we are given a switch $S_{n,m} = (V,E)$ where $V$ is a set of $n$ left ports and $m$ right ports, and $E$ is a set of edges with one left and one right port. Each port $v$ has a corresponding capacity $c_v$ and each edge $e \in E$ has a corresponding set of active rounds $\free{e}$. The problem runs over multiple \textit{rounds} in which some elements of $E$ are scheduled. The objective is to schedule all edges in some round subject to the condition that, for any port $v$, the number of edges that are adjacent to $v$ and are scheduled in a single round is no more than $c_v$. }

\iflong
\begin{remark}
    \emph{
    We note that Time-Constrained Flow Scheduling also generalizes the model in which each edge $e$ has both a release time $r_e$ and a \emph{deadline} $d_e$. In this model, an edge can be scheduled in any round $t$ such that $r_e \le t \le d_e$. Thus, the approximation result proved below applies in the deadline model as well.
    }
\end{remark}
\fi 

\paragraph{Linear Programming Relaxation.} 
We provide a linear programming relaxation of Time-Constrained Flow Scheduling. Let $T = \{t \in \free{e}\}_e$ be the set of rounds in which some edge can be scheduled.
\begin{align}
    \sum_{e\in \flows{p}} \size{e} x_{e,t} &\le \capacity{p} &\forall p \in P, t \in T
    \label{LP:capacity}
    \\
    \sum_{t \in \free{e}} x_{e,t} &= 1 &\forall e \in F
    \label{LP:request}
    \\
    x_{e,t} &\ge 0 &\forall e \in F, t \in T
    \label{LP:nonneg}
\end{align}
The variable $x_{e,t}$ denotes the fraction of flow $e$ scheduled in round $t$.
Constraint (\ref{LP:capacity}) ensures that the total size of all edges adjacent to a port that are scheduled in a round is no more than the port's capacity. Constraint (\ref{LP:request}) ensures that all edges are scheduled.

\begin{theorem}
    Given an instance $\instance{m,m'}$ of Time-Constrained Flow Scheduling, we can either determine that there is no schedule of $\instance{m,m'}$ or produce a schedule in which the capacity of each port has been increased by $2\size{\max}-1$.
    \label{thm:capacity_increase}
\end{theorem}

To prove Theorem~\ref{thm:capacity_increase}, we invoke the following lemma which is proved in  \cite{Karp87globalwire}.
\iflong The lemma implies that, for any solution to the LP, there is a rounded solution where the difference in values is bounded by the sum of positive or negative coefficients for any given variable in \program{\ref{LP:capacity}}{\ref{LP:nonneg}}. \fi

\begin{lemma}[Theorem 3 in \cite{Karp87globalwire}]
    Let $\varmatrix{A}$ be a real-valued $r \times s$ matrix, let $\varmatrix{x}$ be a real-valued $s$-vector, let $\varmatrix{b}$ be a real-valued $r$-vector such that $\varmatrix{Ax} = \varmatrix{b}$, and let $\Delta$ be a positive real number such that in every column of $\varmatrix{A}$ we have \begin{enumerate*}[label=(\alph*)]
        \item the sum of the positive elements is at most $\Delta$ and
        \label{rounding:pos}
        \item the sum of the negative elements is at least $- \Delta$.
        \label{rounding:neg}
        Then we can compute an integral $s$-vector $\varmatrix{\hat{x}}$ such that
        \item for all $i$, $1 \le i \le s$, either $\varmatrix{\hat{x}}_i = \lfloor \varmatrix{x}_i \rfloor$ or $\varmatrix{\hat{x}}_i = \lceil \varmatrix{x}_i \rceil$, and
        \label{rounding:integral}
        \item $\varmatrix{A\hat{x}} = \varmatrix{\hat{b}}$, where $\varmatrix{\hat{b}}_i - \varmatrix{b}_i < \Delta$ for $1 \le i \le r$. In the case that all entries in $\varmatrix{A}$ are integers, then a stronger bound applies: $\varmatrix{\hat{b}}_i - \lceil \varmatrix{b}_i \rceil \le \Delta - 1$.
        \label{rounding:difference}
    \end{enumerate*}
\label{lem:rounding}
\end{lemma}

\begin{proof}[Proof of Theorem \ref{thm:capacity_increase}]
    We first show that LP is a valid relaxation of Time-Constrained Flow Scheduling. We convert a schedule $\schedule{}$ of an arbitrary instance of Time-Constrained Flow Scheduling into a feasible LP solution. For each flow $e$, if $\schedule{}$ schedules $e$ in round $t$, then we set $x_{e,t}$ to 1 and set it to 0 otherwise. By the port capacity restrictions on $\schedule{}$, we have that Constraint~(\ref{LP:capacity}) is satisfied. Also, since all edges must be scheduled in some round of $\schedule{}$, we have that Constraint~(\ref{LP:request}) is satisfied. Constraint~(\ref{LP:nonneg}) is trivially satisfied. 
    
    We now rewrite \program{\ref{LP:capacity}}{\ref{LP:nonneg}} in matrix form as $\varmatrix{A}_{LP} \varmatrix{x} = \varmatrix{b}_{LP}$ with the use of slack variables. Then, for a given instance of Time-Constrained Flow Scheduling, we solve the program. This either outputs that there is no solution or produces a solution vector $\varmatrix{x}^*$. In the former case, we use the fact that LP is a valid relaxation of Time-Constrained Flow Scheduling to determine that there is no feasible solution to the given instance. 
    If the LP solver provides a solution vector $\varmatrix{x}^*$, we rewrite $\varmatrix{A}_{LP}$ and $\varmatrix{b}_{LP}$ as follows. Let $\size{\max} = \max_{e \in F} \{\size{e}\}$. Let $\varmatrix{A}$ and $\varmatrix{b}$ be identical to $\varmatrix{A}_{LP}$ and $\varmatrix{b}_{LP}$ except that all rows corresponding to Constraint (\ref{LP:nonneg}) have been removed, and all values in rows corresponding to constraint (\ref{LP:request}) have been multiplied by $\size{\max}$ and made negative.
    Let $\Delta = \size{\max}$. 
    
    We show that $\varmatrix{A}, \varmatrix{x}, \varmatrix{b}$, and $\Delta$ satisfy the conditions of Lemma~\ref{lem:rounding}.
    By construction we have $\varmatrix{A} \varmatrix{x} = \varmatrix{b}$.
    Let the columns of $\varmatrix{A}$ be indexed by $t \in [T]$.
    Let $pq = e \in F$ be a flow in the given problem instance. Constraint (\ref{LP:capacity}) entails that the coefficient $\size{pq}$ would occur twice in a single column $t$: once for $p$ and once for $q$. Similarly, Constraint (\ref{LP:request}) guarantees that $-2\size{\max}$ occurs once in each column $t$ for $e$. So, conditions \ref{rounding:pos} and \ref{rounding:neg} are satisfied for $\varmatrix{A}, \varmatrix{b},\varmatrix{x}$, and $\Delta$.
    
    Lemma \ref{lem:rounding}, therefore, entails the existence of matrices $\varmatrix{\hat{b}},$ and $\varmatrix{\hat{x}}$ that have properties \ref{rounding:integral} and \ref{rounding:difference}. Property \ref{rounding:integral} entails that all values in $\varmatrix{\hat{x}}$ are integral.
    Since all elements of $\varmatrix{A}$ are integral, we have that all elements of $\varmatrix{\hat{b}}$ are integral as well. Property \ref{rounding:difference} entails that the difference between the values in $\varmatrix{\hat{b}}$ and $\varmatrix{b}$ is strictly less than $2\size{\max}$ and so is at most $2\size{\max} - 1$. 
    
    Recall that all elements of $\varmatrix{A}$ corresponding to Constraint~(\ref{LP:request}) have been multiplied by $2\size{\max}$. Therefore, by dividing these values by $2\size{\max}$, we get a matrix $\varmatrix{b}'$ such that the difference in values $\varmatrix{b}'$ and $\varmatrix{b}$ corresponding to Constraint~(\ref{LP:request}) are strictly less than 1. Since all values are integral, this entail that the difference is 0. 
    \ignore{
    For all elements of $\varmatrix{\hat{b}}$ corresponding to Constraint~(\ref{LP:request}), we can infer the difference is 0. This is because we have doubled all elements of $\varmatrix{A}$ for these constraints, so halving these elements gives new integral values for $\varmatrix{b}$ and $\varmatrix{\hat{b}}$ and this difference must remain integral because $\varmatrix{A}$ and $\varmatrix{\hat{x}}$ are integral.
    }
    Therefore, the schedule given by $\varmatrix{\hat{x}}$ satisfies Constraint~(\ref{LP:request}), and all values corresponding to Constraint~(\ref{LP:capacity}) are off by at most $2\size{\max}-1$. Therefore, if we increase the capacity of each port by $2\size{\max}-1$, we can feasibly schedule all edges in their active rounds.
\end{proof}
\iflong
\begin{remark}
    \emph{
    In the setting with unit demand flows, note that Theorem~\ref{thm:capacity_increase} provides a \textit{tight} approximation: by proof of Theorem~\ref{thm:np_hardness}, it is \textbf{NP}-hard to find an optimal schedule when minimizing maximum makespan even with unit capacities, and increasing capacities by 1 is the smallest change that can be made. 
    }
\end{remark}
\fi

\section{Online Flow Scheduling}
\label{sec:online}
We next consider a natural online version of flow scheduling, in which the sequence of flow requests is not available in advance; the scheduler learns about a request only at the request's release time.  We use the standard framework of competitive analysis, and present some preliminary theoretical results in Section~\ref{sec:onlineprelim}, and experimental results in Section~\ref{sec:experiments}. 

 
    \subsection{Preliminary Theoretical Results}
    \label{sec:onlineprelim} 
    
In this section, we establish several preliminary results for flow scheduling in the online setting. We first describe two lower bounds on the quality of any online approximation for both the average response time and maximum response time objectives. We then provide an online approximation for maximum response time that uses our offline algorithm, described above, as a subroutine. 

\iflong
Recent work in~\cite{dinitz+moseley:schedulingflowsreconfig} for reconfigurable networks, when applied to our model, implies that, for any positive integer $c$, there exists a  $O(1/c^2)$-competitive algorithm for average response time, assuming that the port capacities of the algorithm are $(2 + c)$ times that of the optimal.  The following lemma shows that there is no online algorithm with a bounded competitive ratio for average response time without resource augmentation.
\fi
\iflong
We present a proof for completeness. 
\fi
\begin{lemma}[\cite{kulkarni:flow.lower}]
\label{lem:online_average_lower}
For any $M$, there is an instance $I$ of flow scheduling such that the average response time of the schedule produced by any online algorithm on $I$ is at least $M$ times the average response time of the optimal schedule of $I$.
\end{lemma}
\iflong

\ifsigconf
\begin{figure}
    \centering
    \includegraphics{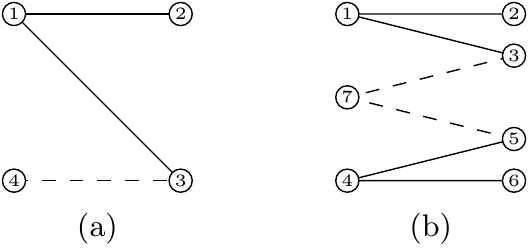}
    \caption{Lower bound constructions. All edges have unit demand and all ports have unit capacity. in (a), the solid flows arrive in each round $1,\ldots,T$ and the dashed flow arrives in each round $T+1,\ldots,M$ for $M \gg T$. In (b), the solid flows arrive in round 1 and the dashed flows arrive in round 2.}
    \label{fig:online_gaps}
\end{figure}
\else
\begin{figure}
    \centering
    \includegraphics[scale=1.2]{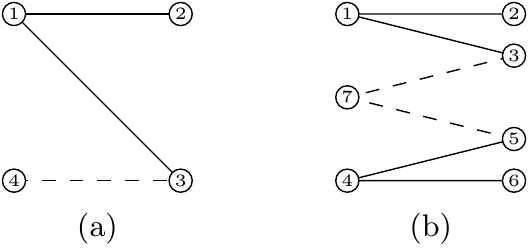}
    \caption{Lower bound constructions. All edges have unit demand and all ports have unit capacity. in (a), the solid flows arrive in each round $1,\ldots,T$ and the dashed flow arrives in each round $T+1,\ldots,M$ for $M \gg T$. In (b), the solid flows arrive in round 1 and the dashed flows arrive in round 2.}
    \label{fig:online_gaps}
\end{figure}
\fi

\begin{proof}
    Let $\mathcal{A}$ be any online algorithm for flow scheduling and consider the problem instance given in Figure~\ref{fig:online_gaps}(a). After the first $T$ rounds, $\mathcal{A}$ will have at least $T/2$ flows remaining adjacent to either port 2 or 3. We suppose without loss of generality that it is  port 3. Let $U$ be the set of flows that are released after round $T$. In this case,
    \begin{align*}
        \sum_e \response{e} &\ge \sum_{e} \completion{e} - \sum_e \release{e} \ge MT - T^2/4.
    \end{align*}
    On the other hand, the optimal executes all flows $(1,3)$ in the first $T$ rounds and then flows $(1,2)$ in parallel with flows $(4,3)$ over the next $T$ rounds, resulting in a total response time $\le 2T$.  The desired claim follows since $M$ can be made arbitrarily large.
\end{proof}
\fi
The following lemma establishes a lower bound for maximum response time, using an argument similar to~\cite{jia+jin+etal:onlineschedulingWAN}.

\begin{lemma}
    There is an instance $I$ of flow scheduling such that the maximum response time of the schedule produced by any online algorithm on $I$ is at least $3/2$ times the optimal for $I$.
\end{lemma}
\iflong
\begin{proof}
    Consider the instance of flow scheduling given in Figure~\ref{fig:online_gaps}. We given an optimal schedule. Schedule flows $(1,3),(4,5)$ in the first round, then $(1,2),(7,3),(4,6)$ in the second, and $(7,5)$ in the third, resulting in a maximum response time of 2.
    Now consider any online algorithm $\mathcal{A}$. After round 1, $\mathcal{A}$ will leave two flows unscheduled. We suppose without loss of generality that these are $(1,3)$ and $(4,5)$. Upon arrival of the dashed flows, $\mathcal{A}$ either schedules $(1,3),(4,5)$ or one of $(7,3)$ or $(7,5)$. In either case, the maximum response time is 3.
\end{proof}
\fi
\begin{lemma}
    There is an online algorithm, which computes a schedule for any given instance $I$, with maximum response time at most double that of the optimal schedule of $I$, and where the  capacity of each port $p$ has been increased to $2(\capacity{p} + 2\size{\max} - 1)$ where $\size{\max} = \max_e \{\size{e}\}$.
    \label{res:onlinemaxapprox}
\end{lemma}

\ifsigconf
\else
\begin{figure}
    \centering
    \includegraphics[]{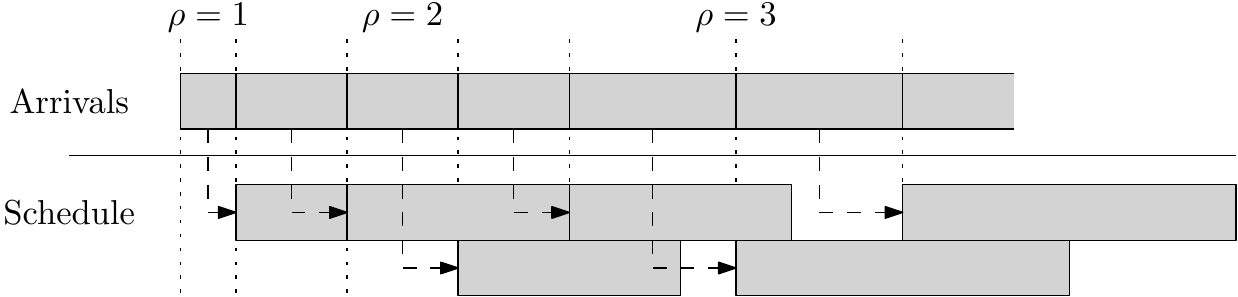}
    \caption{Gray boxes represent batches of flows. At each interval of the guessed max response time $\response{}$, those flows in the previous batch are scheduled starting in the current round. Note that at most two boxes ever overlap.}
    \label{fig:onlinemaxresponse}
\end{figure}
\fi

\ignore{
 Let $\mathcal{A}$ be the offline maximum response time algorithm defined above. Let $\mathcal{A}(B)$ be the schedule produced by running $\mathcal{A}$ on the set of flows $B$. The the online algorithm is given as Algorith~\ref{alg:onlinemaxapprox}.
\begin{algorithm}[h]
    $\tau,t \leftarrow 1$\;
    $U \leftarrow \emptyset$\;
    \While{true}{
        $B \leftarrow \{e: \release{e} = t\}$\; 
        \uIf{$t = i\tau$ for some $i$}{
            \uIf{the max response time of $\mathcal{A}(B) \le 2^{\tau}$}{
                schedule all of $B$ by shifting $\mathcal{A}(B)$ to start at time $t$ and reset $U \leftarrow \emptyset$
            }
            \Else{
            $U \leftarrow U \cup B$\;
            $\tau \leftarrow \tau + 1$}
        }
        \Else{$U \leftarrow U \cup B$\;}
        $t \leftarrow t+1$\;
    }
    \caption{Online \fsmrt \ Approximation}
    \label{alg:onlinemaxapprox}
\end{algorithm}
}
Our online algorithm $\mathcal{A}_{MRT}$ is depicted in Figure~\ref{fig:onlinemaxresponse}. We informally define it here.
In each round $t$, check if $t$ is an integral value of the guessed maximum response time $\response{}$. If so, use the offline algorithm to check if all flows which arrived in the previous $\response{}$ rounds can be scheduled with maximum response time $\response{}$. If so, schedule them according to the offline algorithm starting in round $t$. Otherwise, increase the guessed $\response{}$ by one. 
\iflong

\begin{proof}[Proof of Lemma~\ref{res:onlinemaxapprox}]
    We first show that the maximum response time $\response{}$ of $\mathcal{A}_{MRT}$ is at most double the optimal maximum response time $\response{}^*$. In each round $t$ where $\mathcal{A}_{MRT}$ schedules a batch of flows,  it uses the scheduled derived from the offline algorithm, which we know has optimal response time. However, because the flows are batched by $\response{}$ lengths of time, they may be scheduled up to $2\response{}$ time from the time when they first arrived.
    
    We now show that $\mathcal{A}_{MRT}$ uses up to $2(\capacity{p}+2\size{\max}-1)$ capacity on each port $p$. Consider Figure~\ref{fig:onlinemaxresponse}. If we can show that at most two batches ever overlap, then the result follows from the fact that the offline solution uses no more than $\capacity{p}+2\size{\max}-1$ capacity on each port $p$. So suppose that three batches overlap in some round $t$. We further suppose that $t$ is the earliest round in which three batches overlap. Let $t_1$ be the round in which the first batch started and $t_1'$ the round in which it ended. Let $t_2$ be the rounds in which the second batch started. Since the guessed $\response{}'$ is monotonically nondecreasing, we have $t - t_2 \ge t_2 - t_1$. Since the max response time within any batch is at most double the guessed response time, we also have $t_1' < t_1 + 2(t_2 - t_1)$. Combining these inequalities gives $t_1' < t$, which contradicts our assumption.
\end{proof}
\fi
    
    \subsection{Experimental Results}
    \label{sec:experiments}
    In this section, we describe experiments conducted to evaluate the practical performance of several natural heuristics in on-line scheduling flows over a switch. In these experiments, we measure the average ($\frac{1}{n}\sum \rho_i$) and maximum ($\max \rho_i$) response times of flows. As noted before, the latter objective is predictive of the quality of service (QoS) as perceived by the user: by minimizing the maximum response time, we ensure that no job takes too long to complete. One must keep in mind, however, that optimizing for maximum response time may come at the cost of increased average response time which becomes more relevant when users submit batch jobs.

In the case of average response time, we compare the performance of these heuristics to the optimal value of the linear program (\ref{eq:lp_gk_s})-(\ref{eq:lp_gk_f}) presented in \S\ref{sec:approxmaxresponse}. On the other hand, in the case of maximum response time, we compare the performance of these heuristics to the optimal value of the linear program (\ref{LP:capacity})-(\ref{LP:nonneg}) presented in \S\ref{sec:linear_prog_approach}. Since these LPs give lower bounds on the optimal values of any schedule, they provides us with bases for evaluating the heuristics.

\subsubsection{Methodology}
Packet-level simulators, such as \textsc{ns2}, are not suitable for flow simulation due to the large number of packets generated by each flow which makes the model infeasible. Hence, we have developed an in-house simulator for online flow scheduling of flows over a non-blocking switch.

Specifically, we use a $150 \times 150$ switch with unit port capacities. This switch models a 3000-machine cluster with 150 racks and a total bisection bandwidth of 300Gbps. Thus, each port has a capacity of 1Gbps or 128MBps. Moreover, by setting each time unit to be 1/128 second, each port has a capacity of 1MB per time unit. 

Our simulator maintains a $(150, 150)$ bipartite graph $G_t$ throughout the simulation, where $t$ denotes the time step. The edges in $G_t$ consist of those edges (flows) released at time $t$ plus the ones remaining from previous steps. In other words $E(G_t)$ is the set of released edges waiting to be scheduled. Any heuristic can be plugged in to extract a bipartite matching $M_t\subseteq E(G_t)$. Edges in $M_t$ are assigned to run in time window $t$ to $t+1$. Note that the edges waiting at a particular port form an open queue in the sense that any edge can be selected to run (as opposed to the edge at the front being the only available one).

In each instance of the experiment, flows are generated randomly controlled by two parameters $M$ the average number of flows released per time unit, and $T$ the number of steps during which the flows are generated. More precisely, for each time unit $t=0,.., T-1$, a Poisson distribution of mean $M$ is used to generate flows released at time $t$. For each such flow, an input port and an output port is selected uniformly at random. Note that $M=150$ means that at each port, on average, there is one new flow per time step. Similarly, the average number of new flows per port is 2 and 4 for $M=300$ and $M=600$ respectively.

In our experiments, we compare the following three heuristics.
\begin{mylowitemize} 
\item \textbf{\textsc{MaxCard}}: at every step a matching of maximum cardinality is extracted from $G_t$. This heuristic is guaranteed to keep the largest number of ports busy during each step. We expect a good performance for $\frac{1}{n}\sum \rho_i$ since port utilization is kept at its max, but not for $\max \rho_i$ since it does not distinguish between edges.
\item \textbf{\textsc{MinRTime}}: at every step $t$, each edge $e$ gets assigned a weight equal to $t-r_e$, where $r_e$ is the edge's (flow's) release time. Next, a matching of maximum weight is extracted from $G_t$, where the weight of an edge is the length of time since its release. We expect a good performance for $\max \rho_i$ since the longer an edge has been waiting the higher is its priority. On the other hand, $\frac{1}{n}\sum \rho_i$ may be high due to sub-optimal port utilization.
\item \textbf{\textsc{MaxWeight}}: at every step, each edge gets assigned a weight equal to the sum of queue sizes at its two endpoints. In other words, the weight of an edge is the number of edges incident to its endpoints. Next, a matching of maximum weight is extracted from $G_t$.  Note that the queue size at a port $p$ is the number of released but unscheduled edges having $p$ as an endpoint. We expect this heuristic to perform well for both objectives.
\end{mylowitemize}

Simulations are performed for various values of $M$ and $T$. Specifically, we fix $M\in\{50, 100, 150, 300, 600\}$ and run the simulator for $T\in\{10, 12, 14, 16, 18, 20, 40, 60, 80, 100\}$. Each result is the average of 10 tries. The linear programs are solved only for $T\in\{10, 12, 14, 16, 18, 20\}$ to avoid prohibitively long execution times: even for $M=600$, and $T=20$, each run takes more than 3 hours on an Intel Core-i7 6700HQ machine with 16GB of RAM.

\subsubsection{Implementation}
We implemented the simulator and its tools in C++. We use Lemon 1.3.1 library for various graph algorithms such as traversals and matchings. The \texttt{default\_random\_engine} was used for the distributions. The linear program is modelled and solved using Gurobi 8.1.  In the case of maximum response time, we used a binary-search scheme with the linear program in (\ref{LP:capacity})-(\ref{LP:nonneg}) for finding the minimum feasible response time. The starting point of the binary search is set to the best of the three heuristics.

\subsubsection{Performance}
Figure \ref{fig:avgresponsetime_color}, on page \pageref{fig:avgresponsetime_color}, shows our findings for average response time. The results are compared against the optimal value of the linear program (\ref{eq:lp_gk_s})-(\ref{eq:lp_gk_f}) which provides a lower bound on the optimal average response time. As predicted, overall, \textsc{MaxWeight} and \textsc{MinRTime} are the best and the worst heuristic respectively. However, as the average number of incoming flows $M$ (and hence the congestion) grows, they start to perform very similarly. Curiously, in every scenario, the performance of the the heuristics is within a factor 2 of the linear program. Moreover, the gap seems to close for larger values of $M$.

Figure \ref{fig:maxresponsetime_color}, on page \pageref{fig:maxresponsetime_color}, shows the results for maximum response time. Again, the findings confirm our initial intuition. In particular, \textsc{MinRTime} has consistently the best performance (it almost matches the LP lower bound in some cases). On the other hand, \textsc{MaxWeight} is the worst of the three. Again, all heuristics are always within a factor 2.5 of the LP. Unlike the average case, the gap between the heuristics seems to grow with $M$.

Our conclusion is that \textsc{MaxCard} and \textsc{MinRTime} are good choices for minimizing average response time and minimizing maximum response time respectively. \textsc{MaxWeight} takes the middle ground and is thus the best choice (among the three) when it is desirable to keep both average and maximum response times low.

\section{Open Problems}
\label{sec:open}
We have presented approximation algorithms for minimizing response time metrics in flow scheduling over a switch network.  Our work offers a number of directions for future research.

\noindent \textbf{Improved approximation ratios.} For average response time, our algorithm achieves an $O(\log n/c)$-approximation while incurring a $1 + c$ augmentation in capacity, for any given positive integer $c$.  While resource augmentation is necessary for any competitive algorithm in the online setting, does an offline approximation (with say a polylogarithmic approximation ratio) need resource augmentation?  For maximum response time, our algorithm achieves the optimal objective while incurring an increase in capacity by the size of the maximum demand.  An important open problem is to determine whether we need resource augmentation to obtain any reasonable approximation algorithm for maximum response time.  
\iflong
A technical hurdle in achieving a good approximation without resource augmentation is the following intriguing question, which originates from the iterative rounding approach.  What is the maximum response time achievable for a sequence of unit flow requests represented by bipartite graphs $G_1, G_2, \ldots, G_T$ which satisfy the following condition: for any interval $I$ and any port $v$, the sum, over all $i$ in $I$, of the degrees of $v$ in $G_i$ is at most $|I| + 1$?  That is, in the preceding sub-class of instances, all the requests can be satisfied with response time of 1, assuming an absolutely minimal resource augmentation (of plus 1).  Without any capacity augmentation, can every request be satisfied with a constant response time?  An affirmative answer to this question will likely lead to a compelling approximation algorithm for response time metrics.
\fi 

\noindent \textbf{Competitive online algorithms.}
Our work on online algorithms is preliminary and provides some guidance on heuristics one can use for response-time related metrics.  While we have given a constant-competitive algorithm for maximum response time with constant-factor resource augmentation, the situation with no resource augmentation is unclear.  We plan to conduct a more thorough investigation of online algorithms -- both theoretical and experimental.

\noindent \textbf{Generalizations and beyond worst-case analysis.}
Our work has focused on scheduling flows on switch networks. We would like to extend our research to a broader class of datacenter networks (e.g., trees, fat-trees, more general networks) and more general types of flows (e.g., co-flows).  We would also like to study the problems posed in a model that includes some information about the distribution of input instances that may be available from practical applications.  This would be especially useful for the average response time objective, for which no non-trivial competitive ratio is achievable without resource augmentation.  
\section*{Acknowledgments}
This work was partially supported by NSF grant CCF-1909363.  We would like to thank Janardhan Kulkarni for the many discussions on online flow scheduling, and for generously allowing us to include his 
proof of Lemma~\ref{lem:online_average_lower}.

\ifsigconf
\else

\begin{figure*}[p]
    \centering
    \includegraphics[scale=0.28]{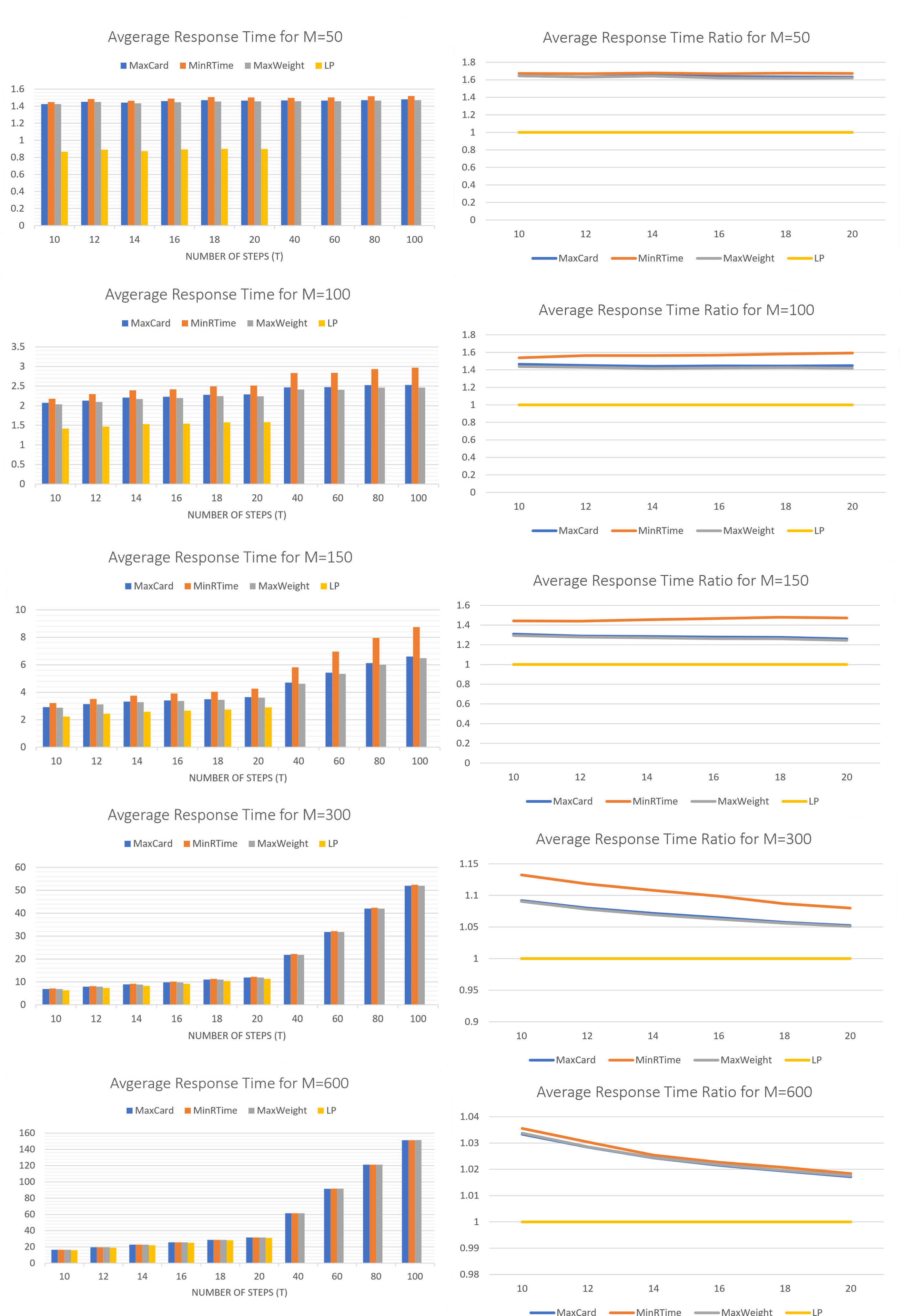}
    \caption{Average response time experimental results.}
    \label{fig:avgresponsetime_color}
\end{figure*}

\begin{figure*}[p]
    \centering
    \includegraphics[scale=0.28]{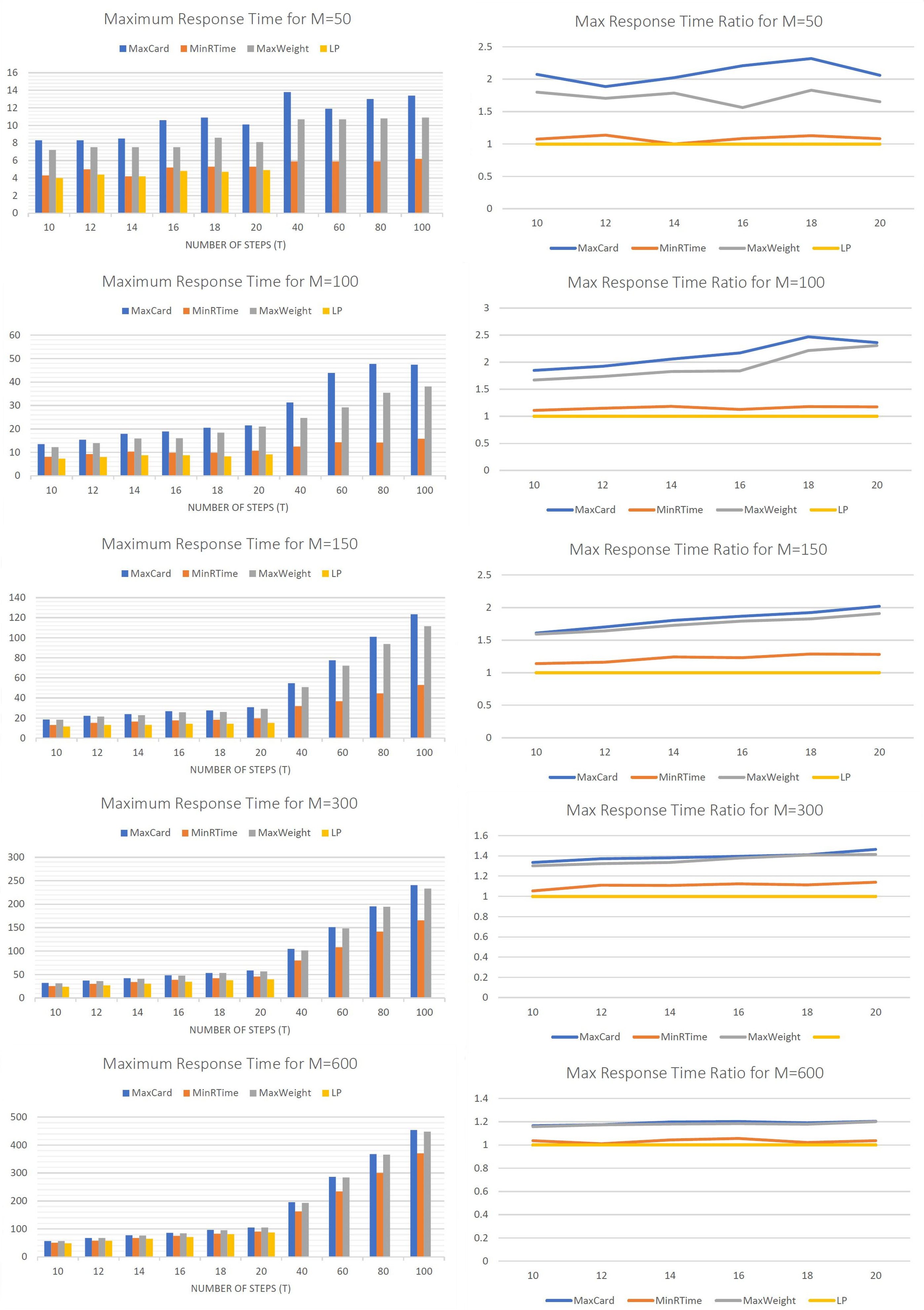}
    \caption{Maximum response time experimental results.}
    \label{fig:maxresponsetime_color}
\end{figure*}

\fi

\newpage
\bibliographystyle{plain}
\bibliography{refs.bib}

\end{document}